\RequirePackage{fix-cm}
\RequirePackage{amsthm}
\documentclass[pdflatex,sn-mathphys-num]{sn-jnl}

\usepackage[utf8]{inputenc}
\usepackage{mathtools,amssymb}
\usepackage{enumitem}
\usepackage{booktabs}
\usepackage{longtable}
\usepackage{array}
\usepackage{aliascnt}

\usepackage{graphicx}
\usepackage{xcolor}
\usepackage[capitalise,noabbrev]{cleveref}

\hypersetup{
  pdftitle={Translation-Invariant Pure States with Exotic Half-Chain Factors},
  pdfauthor={Ruizhi Liu, Jinmin Yi},
  pdfkeywords={quantum spin chain, half-chain algebra, quasifree state,
  type II-infinity factor, type III factor, entanglement entropy}
}

\makeatletter
\renewcommand\normalsize{%
   \@setfontsize\normalsize{11.5bp}{14.5bp}%
   \abovedisplayskip 12\p@ \@plus3\p@ \@minus2\p@
   \abovedisplayshortskip \z@ \@plus3\p@
   \belowdisplayshortskip 6\p@ \@plus3\p@ \@minus3\p@
   \belowdisplayskip \abovedisplayskip
   \let\@listi\@listI}
\normalsize
\renewcommand\small{%
   \@setfontsize\small{11bp}{13.5bp}%
   \abovedisplayskip 10\p@ \@plus3\p@ \@minus3\p@
   \abovedisplayshortskip \z@ \@plus2\p@
   \belowdisplayshortskip 5\p@ \@plus2\p@ \@minus2\p@
   \belowdisplayskip \abovedisplayskip}
\renewcommand\footnotesize{\@setfontsize\footnotesize{9.5bp}{11.5bp}}
\def\abstractfont{\reset@font\fontsize{10.5bp}{13bp}\selectfont
   \leftskip=24pt\rightskip=24pt\parfillskip=0pt plus 1fil}
\def\abstractheadfont{\reset@font\fontsize{10.5bp}{13bp}\bfseries\selectfont\titraggedcenter}
\def\keywordfont{\reset@font\fontsize{10bp}{12bp}\selectfont
   \leftskip=24pt\rightskip=24pt plus0.5fill}
\def\tablecaptionfont{\reset@font\fontsize{10bp}{12bp}\selectfont}
\def\tablebodyfont{\reset@font\fontsize{10bp}{12bp}\selectfont}
\def\tablecolheadfont{\reset@font\fontsize{10bp}{12bp}\selectfont\bfseries\boldmath}
\def\tablefootnotefont{\reset@font\fontsize{9.5bp}{11.5bp}\selectfont}
\def\figurecaptionfont{\reset@font\fontfamily{\rmdefault}\fontsize{10bp}{12bp}\selectfont}
\def\subsubsectionfont{\reset@font\fontsize{12bp}{14bp}\bfseries\selectfont\raggedright\boldmath}
\def\bmheadfont{\reset@font\fontfamily{\rmdefault}\fontsize{11.5bp}{14.5bp}\bfseries\selectfont\raggedright\boldmath}

\makeatother

\numberwithin{equation}{section}
\setlist[enumerate]{leftmargin=2.2em}

\theoremstyle{thmstyleone}
\newtheorem{theorem}{Theorem}[section]
\newaliascnt{proposition}{theorem}\newtheorem{proposition}[proposition]{Proposition}\aliascntresetthe{proposition}
\newaliascnt{lemma}{theorem}\newtheorem{lemma}[lemma]{Lemma}\aliascntresetthe{lemma}
\newaliascnt{corollary}{theorem}\newtheorem{corollary}[corollary]{Corollary}\aliascntresetthe{corollary}
\newaliascnt{conjecture}{theorem}\newtheorem{conjecture}[conjecture]{Conjecture}\aliascntresetthe{conjecture}
\theoremstyle{thmstylethree}
\newaliascnt{definition}{theorem}\newtheorem{definition}[definition]{Definition}\aliascntresetthe{definition}
\newaliascnt{question}{theorem}\newtheorem{question}[question]{Question}\aliascntresetthe{question}
\theoremstyle{thmstyletwo}
\newaliascnt{remark}{theorem}\newtheorem{remark}[remark]{Remark}\aliascntresetthe{remark}

\crefname{theorem}{Theorem}{Theorems}\Crefname{theorem}{Theorem}{Theorems}
\crefname{proposition}{Proposition}{Propositions}\Crefname{proposition}{Proposition}{Propositions}
\crefname{lemma}{Lemma}{Lemmas}\Crefname{lemma}{Lemma}{Lemmas}
\crefname{corollary}{Corollary}{Corollaries}\Crefname{corollary}{Corollary}{Corollaries}
\crefname{conjecture}{Conjecture}{Conjectures}\Crefname{conjecture}{Conjecture}{Conjectures}
\crefname{definition}{Definition}{Definitions}\Crefname{definition}{Definition}{Definitions}
\crefname{question}{Question}{Questions}\Crefname{question}{Question}{Questions}
\crefname{remark}{Remark}{Remarks}\Crefname{remark}{Remark}{Remarks}

\newcommand{\A}{\mathcal A}
\newcommand{\M}{\mathcal M}
\newcommand{\F}{\mathcal F}
\newcommand{\B}{\mathcal B}
\newcommand{\cH}{\mathcal H}
\newcommand{\cS}{\mathcal S}
\newcommand{\hilb}{\mathfrak h}
\newcommand{\Id}{\mathbf 1}
\newcommand{\bbR}{\mathbb R}
\newcommand{\bbC}{\mathbb C}
\newcommand{\bbN}{\mathbb N}
\newcommand{\bbT}{\mathbb T}
\newcommand{\bbD}{\mathbb D}
\newcommand{\z}{\mathbb Z}
\newcommand{\half}{[0,\infty)}
\newcommand{\hb}{\mathrm h_2}
\newcommand{\ovl}{\overline}
\newcommand{\rmd}{\mathrm{d}}
\newcommand{\Hinf}{H^\infty}
\newcommand{\cHinf}{\overline{H^\infty}}
\DeclareMathOperator{\CAR}{CAR}
\DeclareMathOperator{\Tr}{Tr}
\DeclareMathOperator{\rspan}{span}

\DeclareMathOperator{\diag}{diag}
\DeclareMathOperator{\diam}{diam}
\DeclareMathOperator{\Sp}{Sp}
\DeclareMathOperator{\Ran}{Ran}
\DeclareMathOperator{\supp}{supp}
\DeclareMathOperator{\tr}{tr}

\begin{document}

\title[Pure States with Exotic Half-Chain Factors]{Translation-Invariant Pure States with Exotic Half-Chain Factors}

\author[1,2]{\fnm{Ruizhi} \sur{Liu}}
\author[1]{\fnm{Jinmin} \sur{Yi}}

\affil[1]{\orgname{Perimeter Institute for Theoretical Physics},
\orgaddress{\street{31 Caroline Street North}, \city{Waterloo},
\state{Ontario}, \postcode{N2L 2Y5}, \country{Canada}}}

\affil[2]{\orgdiv{Department of Mathematics and Statistics},
\orgname{Dalhousie University}, \orgaddress{\street{6316 Coburg Road},
\city{Halifax}, \state{Nova Scotia}, \postcode{B3H 1Z9}, \country{Canada}}}

\abstract{We construct translation-invariant pure quasifree fermionic states on the one-dimensional lattice whose half-chain von Neumann algebras are hyperfinite factors of type $\mathrm{II}_\infty$ or $\mathrm{III}_\lambda$ for any $0<\lambda\le1$. Via the Jordan--Wigner--Araki transformation, these yield translation-invariant pure spin-chain states with the same half-chain factors, demonstrating that such factors need not be of type I or III. Each state is the unique ground state of a gapless, translation-invariant parent Hamiltonian with superpolynomially decaying interactions. For any fixed factor type, the interval entanglement entropy can exhibit widely varying growth rates, from arbitrarily slow divergence to nearly power-law growth, showing that entropy scaling does not determine the factor type. In contrast, we prove that gauge-invariant quadratic Hamiltonians with finitely many flavours and exponentially decaying hopping exhibit a dichotomy: their ground states either satisfy an area law with a type I half-chain factor, or have logarithmic entropy growth with a type $\mathrm{III}_1$ factor. The counterexample was discovered with AI assistance.}



\maketitle

\section{Introduction}\label{sec:intro}

Quantum entanglement is one of the defining features of quantum theory.
Over the past decades it has become a central concept in quantum
information theory, condensed matter physics and high-energy physics,
where it governs quantum information processing, the characterization of
many-body phases and critical phenomena, and the study of quantum field
theories and holography
\cite{CalabreseCardy2004,Horodecki2009,Amico2008,CalabreseCardy2009,
RyuTakayanagi2006,CalabreseLefevre2008}.

For systems with infinitely many degrees of freedom, such as quantum
fields or lattice systems in the thermodynamic limit, there is no global
tensor-product Hilbert space.  The natural replacement is the
$C^*$-algebra of quasi-local observables together with the von Neumann
algebras generated in its representations
\cite{haag2012local,bratteli2013operator1,bratteli2013operator2,
Naaijkens_2017,Yngvason2005}.  This operator-algebraic viewpoint has long
been used to study entanglement and subsystem structure in infinite
quantum systems \cite{KEYL_2006,matsui2011boundedness}, and it is also the
natural setting for sharp phase transitions and spontaneous symmetry
breaking \cite{Strocchi:2008gsa,Naaijkens_2017}.  More recently, the type
classification of von Neumann algebras has been connected to operational
properties of entanglement such as embezzlement
\cite{vanLuijk2024PRL,vanLuijk2024Embezzlement,vanLuijk2025Pure,
vanluijk2025largescalestructureentanglementquantum}.

Consider the infinite spin chain
\begin{equation}\label{eq:spin-chain}
    \A=\bigotimes_{j\in\z} M_d,
    \qquad
    \A_R:=\A_{\half},
\end{equation}
and a pure state $\omega$ on $\A$.  Although $\omega$ is pure, its
restriction to the right half-chain $\A_R$ is in general mixed.  In the
GNS representation of $\omega$, the von Neumann algebra
\begin{equation}\label{eq:halfchain-alg}
    \M_R(\omega):=\pi_\omega(\A_R)''
\end{equation}
encodes the operator-algebraic structure of the entanglement across the
cut \cite{KEYL_2006}.  The type I case is equivalent to the split property
and recovers the familiar tensor-product picture of the two half-chains,
whereas half-chain algebras that are not of type I describe genuinely
infinite entanglement and are tied to Haag duality and to entanglement
distillation \cite{KEYL_2006,Matsui2001,KEYL_2008}.  Ground states of
gapped finite-range Hamiltonians give type I half-chain algebras
\cite{Hastings2007,Matsui2010,matsui2011boundedness}, while standard
critical models and continuum field theories give factors of type
$\mathrm{III}_1$ \cite{vanLuijk2025,Araki1964,haag2012local}.  Type
$\mathrm{II}_1$ is excluded for translation-invariant pure states
\cite{KEYL_2006,KEYL_2008}.  This is the background of the following
question.

\begin{question}[{\cite[Sec.~5.1]{KEYL_2006}}]
\label{ques:half_chain_type}
Must the half-chain algebra $\M_R(\omega)$ of a translation-invariant pure
state $\omega$ be of type I or of type $\mathrm{III}_{1}$?
\end{question}

This question was raised in \cite{KEYL_2006}.  Subsequently, Mohari
\cite{Mohari2014} claimed to prove that the half-chain algebra of any
translation-invariant pure state must be of type I or of type III, which
would in particular exclude type $\mathrm{II}_\infty$.  However, the
argument intended to rule out type $\mathrm{II}_\infty$ contains a gap,
which we discuss in \cref{sec:discussion}.  The possibility of a type
$\mathrm{II}_\infty$ half-chain algebra has therefore remained open, and it
was recently restated as Question~2 in
\cite{vanluijk2025largescalestructureentanglementquantum}.

On the other hand, entanglement is more commonly measured by the
entanglement entropy.  For a spin chain, an area law, that is, a uniform
bound on the entropy of all finite intervals, forces the half-chain
algebra to be of type I \cite{matsui2011boundedness}.  At the opposite end,
every critical translation-invariant pure quasifree state of spinless
fermions has an entropy that grows at least logarithmically with the
length of the interval \cite{Farkas_2005}, and such states have type
$\mathrm{III}_1$ half-chain algebras (\cref{prop:one-species}).  In both
regimes the growth rate of the entropy and the type of the half-chain
algebra move together, which suggests two further questions.

\begin{question}\label{ques:entropy}
Does the growth rate of the entropy alone determine the type of the
half-chain algebra of a translation-invariant pure state?
\end{question}

\begin{question}\label{ques:type_entropy}
If the half-chain algebra of a translation-invariant pure state is of type
$\mathrm{III}_1$, must the entropy of an interval grow at least
logarithmically with its length?
\end{question}

We answer all three questions in the \textbf{negative}.  More precisely, we prove
the following.

\begin{theorem}\label{thm:main}
There is a family $\mathfrak S$ of translation-invariant pure states on
the spin chain $\A=\bigotimes_{j\in\z}M_{16}$ such that:
\begin{enumerate}[label=\textup{(\roman*)}]
    \item for every $0<\lambda\le1$, $\mathfrak S$ contains a state whose
    half-chain algebra is the hyperfinite factor of type
    $\mathrm{III}_\lambda$, and it contains a state whose half-chain
    algebra is the hyperfinite factor of type $\mathrm{II}_\infty$;
    \item every state in $\mathfrak S$ is the unique ground state of a
    translation-invariant, gapless parent Hamiltonian with
    superpolynomially decaying interactions;
    \item the entanglement entropy can be tuned independently of the type:
    for each of the types in \textup{(i)}, $\mathfrak S$ contains states
    whose entropy grows like $\log\ell$, like $\sqrt{\log\ell}$ or
    $\log\log\ell$, essentially like a power of $\ell$, or more slowly
    than any prescribed divergent function of the interval length $\ell$.
\end{enumerate}
\end{theorem}

In particular, \cref{thm:main}\textup{(i)} produces a translation-invariant
pure spin-chain state with a type $\mathrm{II}_\infty$ half-chain algebra.
This is the announced counterexample to the type I / type III dichotomy,
and it answers \cref{ques:half_chain_type}.  Statement (iii) answers
\cref{ques:entropy,ques:type_entropy}: states of different types can share
exactly the same entropy growth, and type $\mathrm{III}_1$ states can have
an entropy that diverges arbitrarily slowly.

Our construction has three tunable parameters;
\cref{tab:entropy-type-summary} lists representative choices.
\begin{enumerate}[label=\textup{(\arabic*)}]
    \item The number $\nu$ of fermion species.\footnote{We use the words
    species, flavour and fermion spin interchangeably.}
    \item The zeros $0\le a_n<1$ of a Blaschke product, which acts as an
    infinite sequence of ``soft shifts'' (\cref{sec:blaschke}).  The
    numbers $\delta_n:=1-a_n$ set the length scales of the correlations
    and hence the entropy growth, but \emph{not} the type of the
    half-chain algebra.
    \item The mixing ratios $\boldsymbol\kappa=(\kappa_1,\dots,\kappa_{\nu/2})$,
    which determine the type of the half-chain algebra but not the growth
    rate of the entropy.
\end{enumerate}
In particular, by choosing the Blaschke zeros appropriately one obtains an
arbitrarily slow divergence of the entropy while keeping the type of the
half-chain algebra unchanged (\cref{prop:slow}).

\begin{table}[htbp]
\centering
\caption{Representative choices of Blaschke zeros and mixing ratios, with
the resulting half-chain type and entropy growth.  The type depends only
on $\boldsymbol\kappa$, the growth rate only on $\{\delta_n\}$; here
$0<\kappa<1$ and $\alpha>1$.}
\label{tab:entropy-type-summary}
\begin{tabular}{@{}llll@{}}
\toprule
$\delta_{n}=1-a_n$ & Mixing ratios $\boldsymbol\kappa$ & Half-chain type
& Entropy growth \\
\midrule
$2^{-n}$ & $(1)$ & $\mathrm{II}_\infty$ & $\log \ell$ \\
$2^{-n}$ & $(\kappa)$ & $\mathrm{III}_\kappa$ & $\log \ell$ \\
$2^{-n}$ & $(e^{-1},e^{-\sqrt2})$ & $\mathrm{III}_1$ & $\log \ell$
\\[2pt]
$2^{-n^2}$ & $(1)$, $(\kappa)$, $(e^{-1},e^{-\sqrt2})$
& $\mathrm{II}_\infty$, $\mathrm{III}_\kappa$, $\mathrm{III}_1$
& $\sqrt{\log \ell}$ \\[2pt]
$2^{-2^n}$ & $(1)$, $(\kappa)$, $(e^{-1},e^{-\sqrt2})$
& $\mathrm{II}_\infty$, $\mathrm{III}_\kappa$, $\mathrm{III}_1$
& $\log\log \ell$ \\[2pt]
$n^{-\alpha}$ & $(1)$, $(\kappa)$, $(e^{-1},e^{-\sqrt2})$
& $\mathrm{II}_\infty$, $\mathrm{III}_\kappa$, $\mathrm{III}_1$
& $\ell^{1/(\alpha+1)}$, up to a factor $\log\ell$ \\
\bottomrule
\end{tabular}
\end{table}

It is generally believed that exponentially local and superpolynomially
local Hamiltonians behave similarly in gapped phases; see, e.g.,
\cite{Kapustin2022Noether}.  Our examples show that their behaviour can
differ drastically in the gapless setting.  For free fermions this
distinction is already visible in momentum space: superpolynomially local
Hamiltonians give rise to smooth Bloch Hamiltonians, whereas exponentially
local ones give rise to analytic Bloch Hamiltonians.  A smooth function
may vanish to infinite order at an isolated point, while a nonzero
analytic function cannot, and this is the mechanism behind our
counterexamples.  Conversely, analyticity forces a sharp dichotomy
(\cref{thm:analytic-dichotomy}).  For a free-fermion Hamiltonian with
finitely many flavours, exponentially decaying hopping and no flat band
at the Fermi energy, either some band crosses the Fermi energy, and then
the half-chain algebra of the ground state is of type $\mathrm{III}_1$ and
its entropy grows logarithmically, or no band crosses, and then the
half-chain algebra is of type I and the ground state obeys an area law.
Frustration-free Hamiltonians of this kind always fall into the second
case.

Logarithmic entanglement growth is common in one-dimensional critical
systems and has been established for broad classes of translation-invariant
quasifree models \cite{Vidal2003Entanglement,Jin2004Quantum,
Keating2004Random,Farkas_2005,Wolf2006Violation,GioevKlich2006,
Its2008Entanglement}.  In particular, Farkas and Zimborás proved a
logarithmic lower bound for every nontrivial pure gauge-invariant
translation-invariant quasifree state of one species \cite{Farkas_2005}.
By contrast, sublogarithmic growth can occur for fermionic chains with
sufficiently long-range couplings \cite{Ares2019Sublogarithmic}.
\Cref{thm:analytic-dichotomy} shows that such behaviour is excluded for
ground states of exponentially local free-fermion Hamiltonians with
finitely many flavours.

We expect this dichotomy to persist beyond free fermions; it seems to be
part of the folklore in the physics community
\cite{CalabreseCardy2004,Jin2004Quantum,CalabreseCardy2009}.  The
following conjecture is closely related to
\cref{ques:half_chain_type,ques:entropy,ques:type_entropy} and remains
open.

\begin{conjecture}[Local-ground-state dichotomy]
\label{conj:short-range-dichotomy}
Let $\Phi$ be a translation-invariant interaction on a one-dimensional spin
chain whose terms decay exponentially with their diameter, and suppose
that $\Phi$ has a unique ground state $\omega$.  Then exactly one of the
following holds: either $\M_R(\omega)$ is of type I and the entropy obeys
an area law, or $\M_R(\omega)$ is of type $\mathrm{III}_1$ and the entropy
of an interval grows at least logarithmically with its length.  If
moreover $\Phi$ is frustration free,\footnote{That is, $\Phi(X)\ge0$ and
$\pi_\omega(\Phi(X))\Omega_\omega=0$ for every $X$, where
$(\pi_\omega,\cH_\omega,\Omega_\omega)$ is the GNS representation of
$\omega$.} the first alternative holds.
\end{conjecture}

The first alternative can occur even when the Hamiltonian is gapless, for
instance at a quadratic band touching.  The frustration-free part of the
conjecture is motivated by recent work
\cite{MasaokaSoejimaWatanabe2024Quadratic,SengokuPoWatanabe2025Quasilocal,
MasaokaSoejimaWatanabe2025Ising,MasaokaSoejimaWatanabe2025Dynamical,
OnoMasaokaWatanabePo2026Free,MasaokaOnoPoWatanabe2026Fermionic} arguing
that a conformal field theory cannot emerge from a frustration-free local
Hamiltonian.  Since lattice models with a conformal scaling limit are
expected to have half-chain algebras of type $\mathrm{III}_1$
\cite{vanLuijk2025,vanluijk2025largescalestructureentanglementquantum},
the conjecture is consistent with that broader picture.
\Cref{thm:analytic-dichotomy} proves it for free fermions without a flat
band at the Fermi energy.

\medskip\noindent\textbf{Organization.}
Our states are translation-invariant pure quasifree states of several
species of fermions, so \cref{sec:prelim} reviews quasifree states,
together with the Toeplitz and Hankel operators that describe their
half-chain restrictions.  \Cref{sec:constructing} first explains why one
species of fermions can only give type I or type $\mathrm{III}_1$, and
then constructs the states of type $\mathrm{II}_\infty$ and
$\mathrm{III}_\lambda$.  \Cref{sec:parent} constructs parent Hamiltonians
and proves the dichotomy for exponentially local free fermions.
\Cref{sec:entropy} estimates the entropy and shows that the mixing ratios
do not affect its growth rate.  \Cref{sec:JWA} transfers everything to
the spin chain via the Jordan--Wigner--Araki transformation and proves
\cref{thm:main}, and \cref{sec:discussion} discusses the results.  The
construction uses classical analysis, namely Hardy spaces, Blaschke
products and Toeplitz operators, which we review in \cref{app:analysis};
facts on quasifree states are collected in \cref{app:quasifree}.

\section{Preliminaries}\label{sec:prelim}

\Cref{tab:notation} lists the notation used throughout the paper, its
meaning, and the place where it first appears. For readers who are already familiar with the background, this section can be skipped safely.

{\footnotesize
\setlength{\tabcolsep}{4pt}
\renewcommand{\arraystretch}{1.15}
\begin{longtable}{@{}>{\raggedright}p{0.22\textwidth}>{\raggedright}p{0.5\textwidth}
>{\raggedright\arraybackslash}p{0.2\textwidth}@{}}
\caption{List of notation.}\label{tab:notation}\\
\toprule
Notation & Meaning & First appearance\\
\midrule
\endfirsthead
\multicolumn{3}{@{}l}{\emph{\Cref{tab:notation}, continued}}\\
\toprule
Notation & Meaning & First appearance\\
\midrule
\endhead
\midrule
\multicolumn{3}{r@{}}{\emph{continued on the next page}}\\
\endfoot
\bottomrule
\endlastfoot
\multicolumn{3}{@{}l}{\emph{General conventions}}\\[2pt]
$\Id$ & identity operator, identity element & \cref{sec:CAR}\\
$\langle\cdot,\cdot\rangle$ & inner product, linear in the second variable
& \eqref{eq:CAR_alg}\\
$\|\cdot\|$ & operator norm & \cref{sec:toeplitz}\\
$\|\cdot\|_p$ & Schatten $p$-norm; $\|\cdot\|_1$ trace norm,
$\|\cdot\|_2$ Hilbert--Schmidt norm & \cref{lem:toeplitz-hankel}\\
$\cS_p$ & Schatten class of operators with $\|\cdot\|_p<\infty$
& \cref{lem:halfchain}\\
$A\asymp B$ & $cB\le A\le CB$ with constants $c,C>0$ independent of the
variables in question & \cref{prop:radial}\\
$\hb(t)$ & binary entropy $-t\log t-(1-t)\log(1-t)$, natural logarithm
& \cref{lem:entropy-bounds}\\
$\lambda$, $\mathrm{III}_\lambda$ & label of a factor of type
$\mathrm{III}_\lambda$ & \cref{thm:main}\\
$\ell$ & length of an interval & \cref{thm:main}\\
$\mathfrak S$ & family of spin-chain states of the main theorem
& \cref{thm:main}\\[4pt]
\multicolumn{3}{@{}l}{\emph{Spin chains and operator algebras}}\\[2pt]
$\A$, $\A_R$ & quasi-local algebra $\bigotimes_{j\in\z}M_d$ of the spin
chain; algebra of the right half-chain $[0,\infty)$
& \eqref{eq:spin-chain}\\
$(\pi_\omega,\cH_\omega,\Omega_\omega)$ & GNS representation of a state
$\omega$ & \eqref{eq:halfchain-alg}\\
$\M_R(\omega)$ & half-chain algebra $\pi_\omega(\A_R)''$
& \eqref{eq:halfchain-alg}\\
$\Phi$ & translation-invariant interaction & \cref{conj:short-range-dichotomy}\\
$S_\omega(\Lambda)$, $S_\omega(\ell)$ & entanglement entropy of a finite set
$\Lambda$, of the interval $\{0,\dots,\ell-1\}$ & \cref{sec:CAR}\\[4pt]
\multicolumn{3}{@{}l}{\emph{Fermions and quasifree states}}\\[2pt]
$\nu$ & number of fermion species per site & \cref{sec:intro}\\
$\hilb$ & one-particle space $\ell^2(\z)\otimes\bbC^\nu$
& \eqref{eq:one-particle}\\
$\CAR(\hilb)$ & CAR algebra over $\hilb$ & \eqref{eq:CAR_alg}\\
$c(f)$, $c^*(f)$, $c_{x,a}$ & annihilation and creation operators;
annihilation operator of species $a$ at site $x$ & \eqref{eq:CAR_alg}\\
$\Theta$ & parity automorphism $c(f)\mapsto-c(f)$; even and odd elements
& \cref{sec:CAR}\\
$\A_{\rm f}(\Lambda)$ & CAR algebra of the modes in $\Lambda\subset\z$
& \cref{sec:CAR}\\
$\omega_C$, $C$ & gauge-invariant quasifree state with covariance $C$
& \eqref{eq:quasifree}\\
$\hilb_R$, $C_R$ & one-particle space of the right half-chain;
right-half covariance $\Pi C\Pi$ & \eqref{eq:CR-def}\\[4pt]
\multicolumn{3}{@{}l}{\emph{Momentum space and Hardy space}}\\[2pt]
$\bbT$, $m$, $z=e^{ik}$ & unit circle (Brillouin zone) with normalized
Haar measure; crystal momentum $k$ & \cref{sec:toeplitz}\\
$\F$ & Fourier transform $\delta_x\mapsto z^x$ & \cref{sec:toeplitz}\\
$M_F$, $\hat F_r$ & multiplication operator by $F$; Fourier coefficients
of $F$ & \eqref{eq:matrix-elements}\\
$H^2$, $\Pi$ & Hardy space; projection onto $H^2\otimes\bbC^n$
& \cref{sec:toeplitz}\\
$T_F$, $H_F$ & Toeplitz and Hankel operators with symbol $F$
& \eqref{eq:toeplitz-hankel-def}\\
$I_\ell$, $Q_\ell$ & interval $\{0,\dots,\ell-1\}$; projection onto its
modes & \cref{sec:toeplitz}\\
$\Hinf$, $\cHinf$ & bounded analytic functions on the unit disk $\bbD$;
their complex conjugates & \cref{app:fourier}\\[4pt]
\multicolumn{3}{@{}l}{\emph{The construction}}\\[2pt]
$a_n$, $\delta_n$ & zeros of the Blaschke product; $\delta_n=1-a_n$
& \cref{sec:intro}\\
$\boldsymbol\kappa$ & vector of mixing ratios $(\kappa_1,\dots,\kappa_m)$
& \cref{sec:intro}\\
$b_a$ & Blaschke factor $(a-z)/(1-az)$ & \eqref{eq:blaschke_factor}\\
$u$ & infinite Blaschke product $\prod_nb_{a_n}$ & \eqref{eq:standing}\\
$K_\theta$ & model space $H^2\ominus\theta H^2$ of an inner function
$\theta$ & \eqref{eq:model-space}\\
$\kappa$, $C_\kappa$ & mixing ratio; covariance of one block
& \eqref{eq:C_kappa}\\
$C_{\boldsymbol\kappa}$, $\omega_{\boldsymbol\kappa}$ & covariance of $m$
blocks; corresponding pure quasifree state & \eqref{eq:2m_flavor}\\
$J_{\boldsymbol\kappa}$, $G_{\boldsymbol\kappa}$ & indices with
$0<\kappa_j<1$; closed group generated by these $\kappa_j$
& \cref{prop:modular-family-types}\\[4pt]
\multicolumn{3}{@{}l}{\emph{Hamiltonians}}\\[2pt]
$h$, $h(z)$, $h_r$ & one-particle Hamiltonian; Bloch Hamiltonian;
hopping amplitude at distance $r$ & \cref{sec:ground-states}\\
$\delta_h$, $\A_{\rm loc}$ & generator of the quasifree dynamics; local
elements of the CAR algebra & \eqref{eq:delta-h}\\
$m_\kappa$, $t_\kappa$ & mass and inter-flavour hopping of the flat-band
Hamiltonian & \eqref{eq:flat_parent}\\
$q$ & mollifier $\exp(-|1-z|^{-2})$ & \eqref{eq:q-def}\\
$\varepsilon_j$, $\varphi_j$ & analytic bands and eigenvectors
& \cref{lem:rellich}\\
$P_{\rm F}$, $\omega_{\rm F}$ & Fermi projection
$\chi_{(-\infty,0)}(h)$; its quasifree state
& \cref{thm:analytic-dichotomy}\\[4pt]
\multicolumn{3}{@{}l}{\emph{Entropy estimates}}\\[2pt]
$S_\kappa(\ell)$ & entropy of one block $\omega_\kappa$ on $I_\ell$
& \cref{sec:entropy}\\
$\gamma_\kappa$, $F_\kappa$ & $4\kappa/(1+\kappa)^2$; entropy of one pair of
modes & \cref{prop:interval-entropy}\\
$T_\ell$, $D_\ell$ & compressed Toeplitz matrix $Q_\ell T_uQ_\ell$;
$\Tr(\Id-T_\ell^*T_\ell)$ & \cref{sec:interval-entropy}\\
$\epsilon_\ell$ & $1-e^{-1/(2\ell)}$ & \cref{prop:radial}\\
$\mathcal N(L)$ & number of $n$ with $\delta_n\ge1/L$ & \cref{sec:lacunary}\\[4pt]
\multicolumn{3}{@{}l}{\emph{Jordan--Wigner--Araki transformation}}\\[2pt]
$\A_{\rm s}$, $\Gamma_x$, $\Theta_{\rm s}$ & spin chain
$\bigotimes_\z M_{2^\nu}$; on-site parity; its grading & \cref{sec:JWA}\\
$\vartheta_I$, $\vartheta$ & local Jordan--Wigner maps; isomorphism of
even parts & \cref{lem:jw-local}\\
$\omega_{\rm s}$ & spin state dual to an even fermionic state
& \cref{def:jw-state}\\
$\M^{\rm f}_{R/L}$, $\M^{\rm s}_{R/L}$ & half-chain algebras of
$\omega_{\rm f}$ and $\omega_{\rm s}$ & \cref{lem:jw-halfchain}\\
$\Phi_{\rm f}$, $\Phi_{\rm s}$, $\delta_{\rm s}$ & fermionic and spin
interactions; spin derivation & \eqref{eq:spin-interaction}\\
$\Sp_\tau(B)$ & Arveson spectrum of $B$ & \cref{app:groundstates}\\
\end{longtable}
}

\subsection{Lattice fermions and quasifree states}\label{sec:CAR}

For $\nu$ species of fermions on $\z$ the one-particle space is
\begin{equation}\label{eq:one-particle}
    \hilb=\ell^2(\z)\otimes\bbC^\nu ,
\end{equation}
with orthonormal basis $\delta_x\otimes e_a$, where $\delta_x$ is a
particle localized at the site $x$ and $e_a$ labels the species (also
called flavour or spin).  The algebra of observables is the CAR algebra
$\CAR(\hilb)$, generated by the annihilation operators $c(f)$, antilinear
in $f$, and the creation operators $c^*(f)=c(f)^*$, subject to
\begin{equation}\label{eq:CAR_alg}
    \{c(f),c(g)\}=0,
    \qquad
    \{c(f),c^*(g)\}=\langle f,g\rangle\Id .
\end{equation}
We write $c_{x,a}:=c(\delta_x\otimes e_a)$, and $\Theta$ for the parity
automorphism $c(f)\mapsto-c(f)$; elements fixed by $\Theta$ are even, those
changed in sign are odd.  For $\Lambda\subset\z$ let
$\A_{\rm f}(\Lambda)=\CAR(\ell^2(\Lambda)\otimes\bbC^\nu)$; for finite
$\Lambda$ this is a full matrix algebra of dimension $2^{2\nu|\Lambda|}$.

For every positive contraction $0\le C\le\Id$ on $\hilb$ there is a unique
gauge-invariant quasifree state $\omega_C$, determined by
\begin{equation}\label{eq:quasifree}
    \omega_C\bigl(c^*(f_n)\cdots c^*(f_1)c(g_1)\cdots c(g_k)\bigr)
    =\delta_{nk}\det\bigl[\langle g_i,Cf_j\rangle\bigr]_{i,j=1}^{n}
\end{equation}
\cite{PowersStormer1970,Araki1970}.  In particular
$\omega_C(c^*(f)c(g))=\langle g,Cf\rangle$; the operator $C$ is the
covariance, or one-particle correlation matrix.  The state $\omega_C$ is
invariant under the gauge automorphisms $c(f)\mapsto c(e^{i\phi}f)$, hence
even, and it is pure if and only if $C$ is a projection.  Its restriction
to $\CAR(\mathfrak k)$, for a closed subspace $\mathfrak k\subset\hilb$ with
projection $P_{\mathfrak k}$, is the quasifree state with covariance
$P_{\mathfrak k}CP_{\mathfrak k}|_{\mathfrak k}$.  These and further facts
are collected in \cref{app:quasifree}.  For a finite set $\Lambda$ we
write $S_\omega(\Lambda)$ for the von Neumann entropy of the density
matrix of $\omega|_{\A_{\rm f}(\Lambda)}$, and
$S_\omega(\ell):=S_\omega(\{0,\dots,\ell-1\})$.

\subsection{Momentum space, Toeplitz and Hankel operators}
\label{sec:toeplitz}

The Fourier transform $\F:\ell^2(\z)\to L^2(\bbT)$, $\delta_x\mapsto z^x$,
is unitary, where $\bbT=\{z\in\bbC:|z|=1\}$ carries the normalized Haar
measure $m$ (\cref{app:fourier}).  It turns the translation
$\delta_x\mapsto\delta_{x+1}$ into multiplication by $z$, and the
translation-invariant bounded operators on $\hilb$ into multiplication
operators $M_F$ by matrix functions $F\in L^\infty(\bbT,M_\nu)$.  With
$\hat F_r:=\int_\bbT F(z)z^{-r}\,\rmd m(z)$ the matrix elements are
\begin{equation}\label{eq:matrix-elements}
    \langle\delta_y\otimes e_a,M_F\,\delta_x\otimes e_b\rangle
    =(\hat F_{y-x})_{ab}.
\end{equation}
The one-particle space $\ell^2(\bbN)\otimes\bbC^\nu$ of the right
half-chain, $\bbN=\{0,1,2,\dots\}$, becomes $H^2\otimes\bbC^\nu$, where
$H^2=\ovl{\rspan}\{1,z,z^2,\dots\}\subset L^2(\bbT)$ is the Hardy space.
For a multiplicity $n\ge1$ let $\Pi$ be the orthogonal projection of
$L^2(\bbT)\otimes\bbC^n$ onto $H^2\otimes\bbC^n$; on $\hilb$ we have
$n=\nu$, and \cref{sec:constructing,sec:entropy} also use the scalar case
$n=1$.  For $F\in L^\infty(\bbT,M_n)$ the Toeplitz and Hankel operators
are
\begin{equation}\label{eq:toeplitz-hankel-def}
    T_F:=\Pi M_F|_{H^2\otimes\bbC^n},
    \qquad
    H_F:=(\Id-\Pi)M_F|_{H^2\otimes\bbC^n}
\end{equation}
\cite{Bottcher2006}.  Both have norm at most $\|F\|_\infty$, and
$T_F^*=T_{F^*}$.  For an interval $I_\ell=\{0,\dots,\ell-1\}$ let $Q_\ell$
be the projection onto $\ell^2(I_\ell)\otimes\bbC^n$.

\begin{lemma}\label{lem:toeplitz-hankel}
Let $F\in L^\infty(\bbT,M_n)$.
\begin{enumerate}[label=\textup{(\roman*)}]
    \item $\|H_F\|_2^2=\sum_{r\ge1}r\,\|\hat F_{-r}\|_2^2$.
    \item With $\mu_\ell(r):=\min\{|r|,\ell\}$,
    \begin{gather*}
    \|(\Id-Q_\ell)M_FQ_\ell\|_2^2=\sum_{r\ne0}\mu_\ell(r)\,\|\hat F_r\|_2^2,
    \\
    \|(\Id-Q_\ell)M_FQ_\ell\|_1\le\sum_{r\ne0}\mu_\ell(r)\,\|\hat F_r\|_1 .
    \end{gather*}
\end{enumerate}
\end{lemma}

\begin{proof}
By
\cref{eq:matrix-elements} the operators consist of the blocks
$\hat F_{y-x}$ with $x\ge0>y$, respectively with $x\in I_\ell$ and
$y\notin I_\ell$.  For fixed $r\ge1$ there are $r$ pairs with $y-x=-r$ in
the first case, and for fixed $r\ne0$ there are $\min\{|r|,\ell\}$ pairs
with $y-x=r$ in the second.  The trace-norm bound follows from the
triangle inequality, each block being an operator of trace norm
$\|\hat F_{y-x}\|_1$.
\end{proof}

\subsection{Half-chain algebras of quasifree states}\label{sec:halfchain}

Let $\omega_C$ be a pure gauge-invariant quasifree state on $\CAR(\hilb)$,
$\hilb_R=\ell^2(\bbN)\otimes\bbC^\nu$, and
\begin{equation}\label{eq:CR-def}
    \M_R(\omega_C):=\pi_{\omega_C}\bigl(\CAR(\hilb_R)\bigr)'',
    \qquad
    C_R:=\Pi C\Pi|_{\hilb_R}.
\end{equation}
If $C=M_{C(z)}$ is translation invariant, then $C_R=T_{C(z)}$.

\begin{lemma}\label{lem:halfchain}
$\M_R(\omega_C)$ is a factor isomorphic to
$\pi_{C_R}(\CAR(\hilb_R))''$, the von Neumann algebra of the GNS
representation of the quasifree state $\omega_{C_R}$.  It is of type I if
and only if $(\Id-\Pi)C\Pi\in\cS_2$; for translation-invariant $C=M_{C(z)}$
this means $\sum_{r\ge1}r\|\hat C_{-r}\|_2^2<\infty$.
\end{lemma}

\begin{proof}
Since $\omega_C$ is even and pure, $\M_R(\omega_C)$ is a factor
\cite[Lemma~6]{vanLuijk2025}.  The projection $p$ onto the closure of
$\pi_{\omega_C}(\CAR(\hilb_R))\Omega$ lies in $\M_R(\omega_C)'$, and
$\M_R(\omega_C)p$ is the von Neumann algebra of the GNS representation of
$\omega_C|_{\CAR(\hilb_R)}=\omega_{C_R}$.  As $\M_R(\omega_C)$ is a factor,
$x\mapsto xp$ is an isomorphism \cite[Prop.~5.5.6]{kadison1997fundamentals1}.
By \cref{prop:qf-basic}, $\omega_{C_R}$ is of type I if and only if
$C_R-C_R^2$ is trace class, and
$C_R-C_R^2=\Pi C(\Id-\Pi)C\Pi|_{\hilb_R}=X^*X$ with $X=(\Id-\Pi)C\Pi$.  The
last statement is \cref{lem:toeplitz-hankel}(i), since
$(\Id-\Pi)C\Pi=H_{C(z)}$.
\end{proof}

\subsection{Entropy of intervals}\label{sec:entropy-prelim}

\begin{lemma}\label{lem:entropy-bounds}
Let $C$ be a projection on $\hilb$, $\Lambda\subset\z$ finite, $Q_\Lambda$
the projection onto $\ell^2(\Lambda)\otimes\bbC^\nu$, and
$X_\Lambda:=(\Id-Q_\Lambda)CQ_\Lambda$.  Then
\begin{equation}\label{eq:entropy-bounds}
    4\log2\,\|X_\Lambda\|_2^2
    \le S_{\omega_C}(\Lambda)
    \le 2\log2\,\|X_\Lambda\|_1 .
\end{equation}
\end{lemma}

\begin{proof}
By \cref{lem:qf-entropy}, $S_{\omega_C}(\Lambda)=\Tr\hb(C_\Lambda)$ with
$C_\Lambda=Q_\Lambda CQ_\Lambda|_{\Ran Q_\Lambda}$.  Since
$C_\Lambda-C_\Lambda^2=Q_\Lambda C(\Id-Q_\Lambda)CQ_\Lambda
=X_\Lambda^*X_\Lambda$, the numbers $\eta(1-\eta)$, $\eta$ an eigenvalue of
$C_\Lambda$, are the squared singular values of $X_\Lambda$.  Topsøe's
inequalities \cite{Topsoe2001}
\begin{equation}\label{eq:topsoe}
    4\log2\;t(1-t)\le\hb(t)\le\log2\,\bigl(4t(1-t)\bigr)^{1/\log4}
    \le2\log2\,\sqrt{t(1-t)}
\end{equation}
for $0\le t\le1$, where the last step uses $1/\log4>1/2$, give the claim.
\end{proof}

\begin{corollary}\label{cor:entropy-diverges}
If $\M_R(\omega_C)$ is not of type I, then $S_{\omega_C}(\ell)\to\infty$.
\end{corollary}

\begin{proof}
Since $Q_\ell\le\Pi$, we have $\|X_{I_\ell}\|_2\ge\|(\Id-\Pi)CQ_\ell\|_2$,
which increases to $\|(\Id-\Pi)C\Pi\|_2=\infty$ by \cref{lem:halfchain}.
\end{proof}

\section{Constructing the family}\label{sec:constructing}

\subsection{One species of fermions}\label{sec:spinless_fermion}

\begin{proposition}\label{prop:one-species}
Let $\omega_C$ be a translation-invariant pure gauge-invariant quasifree
state of one species of fermions.  Then $C=M_{\chi_E}$ for a measurable
set $E\subset\bbT$, the filled Fermi sea.  If $E$ or $\bbT\setminus E$ is
a null set, $\M_R(\omega_C)$ is of type I; otherwise it is of type
$\mathrm{III}_1$.
\end{proposition}

\begin{proof}
Translation invariance means that $C$ commutes with $M_z$.  Put
$\chi:=C\Id\in L^2(\bbT)$; then $Cp=CM_p\Id=M_pC\Id=\chi p$ for every
trigonometric polynomial $p$, so $\|\chi p\|_2\le\|p\|_2$, whence
$\chi\in L^\infty$ and $C=M_\chi$ by density.  As $C$ is a projection,
$\chi^2=\chi$, so $\chi=\chi_E$.  Then $C_R=T_{\chi_E}$.  If $E$ or its
complement is null, $C_R\in\{0,\Id\}$ is a projection, $\omega_{C_R}$ is
pure and the factor is of type I by \cref{lem:halfchain}.  Otherwise the
spectrum of the self-adjoint Toeplitz operator $T_{\chi_E}$ is the interval
between the essential infimum and supremum of $\chi_E$, that is $[0,1]$
\cite{HartmanWintner1954,Bottcher2006}, and a quasifree state whose
covariance has a nondegenerate interval in its spectrum generates a factor
of type $\mathrm{III}_1$ \cite[Lemma~26]{vanLuijk2025}.
\end{proof}

\begin{figure}[htbp]
    \centering
    \includegraphics[width=\textwidth]{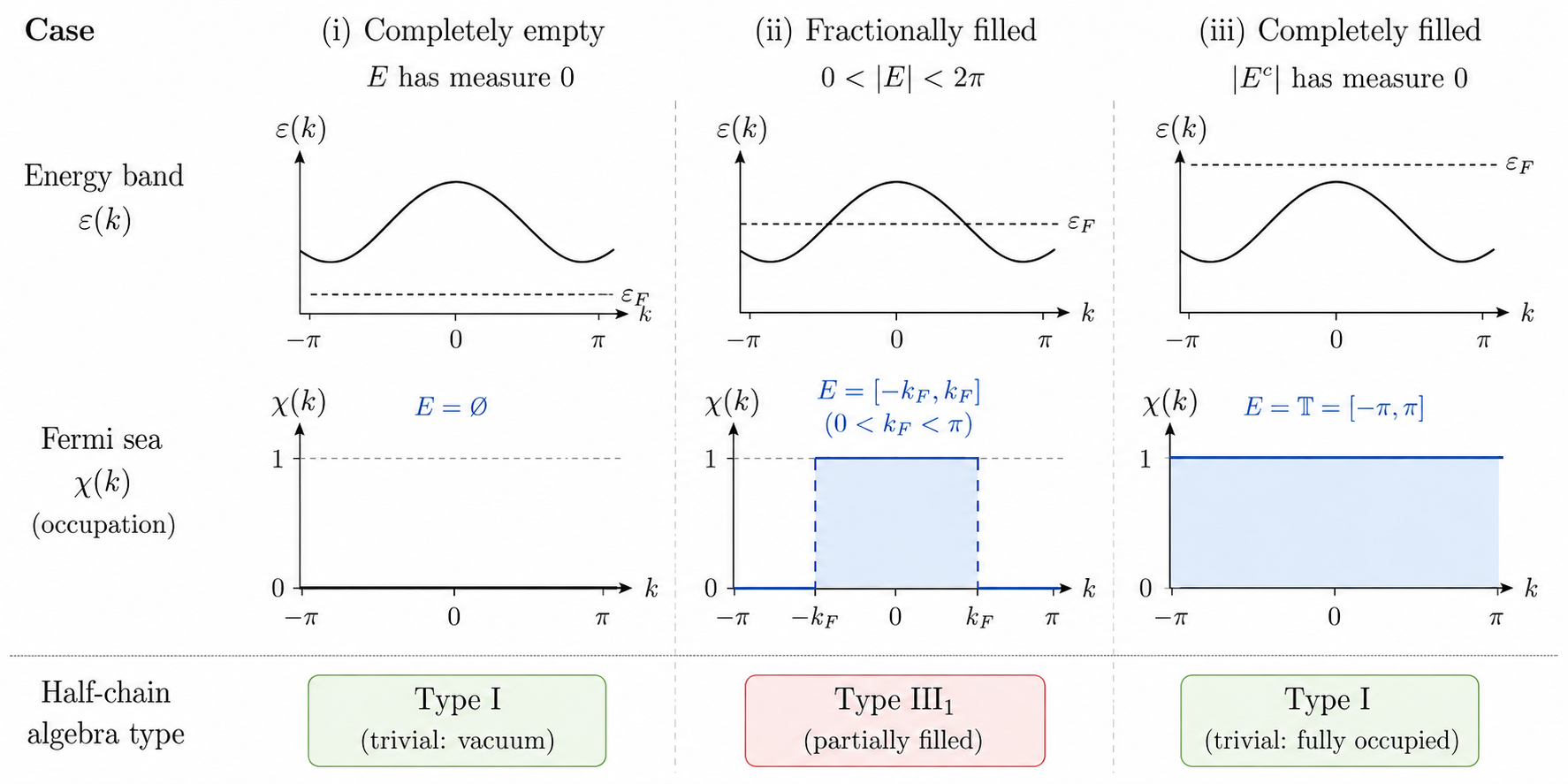}
    \caption{Fermi seas of a translation-invariant pure quasifree state of
    one fermion species.  An empty or completely filled band gives a type
    I half-chain algebra; every partially filled band, with Fermi sea $E$
    and its complement of positive measure, gives type $\mathrm{III}_1$.}
    \label{fig:spinless-fermion-band}
\end{figure}

One species therefore yields only type I and type $\mathrm{III}_1$.  Other
types require internal degrees of freedom, and the flavours must be
correlated: if the covariance is block diagonal in the flavour index, so is
the right-half covariance, the half-chain factor is the tensor product of
the single-species factors up to isomorphism
\cite{PowersStormer1970,ArakiWoods1968}, and the type remains I or
$\mathrm{III}_1$.

\subsection{Blaschke products as soft shifts}\label{sec:blaschke}

The second ingredient is the Blaschke product \cite{Duren1970,Garnett2007};
the mathematical facts are collected in \cref{app:blaschke}.  For
$0\le a<1$ put
\begin{equation}\label{eq:blaschke_factor}
    b_{a}(z)=\frac{a-z}{1-az}
    =a-(1-a^{2})\sum_{n=1}^{\infty}a^{n-1}z^{n}.
\end{equation}
A particle localized at the site $0$ is the constant function $1$ in
momentum space, and the unitary $M_{b_a}$ maps it to
\begin{equation}
    M_{b_a}\delta_0=a\,\delta_0-(1-a^{2})\sum_{n=1}^{\infty}a^{n-1}\delta_n .
\end{equation}
Part of the amplitude stays at the origin and the rest spreads to the
right with an exponential tail of decay length $\xi_a=1/|\log a|$
(\cref{fig:onezero_soft_shift}).  For $a=0$ we get $b_0(z)=-z$, the
one-site translation up to a sign; for $a$ close to $1$ the tail carries
little weight but is long.  A Blaschke factor is therefore a soft shift.

\begin{figure}[htbp]
    \centering
    \includegraphics[width=0.9\linewidth]{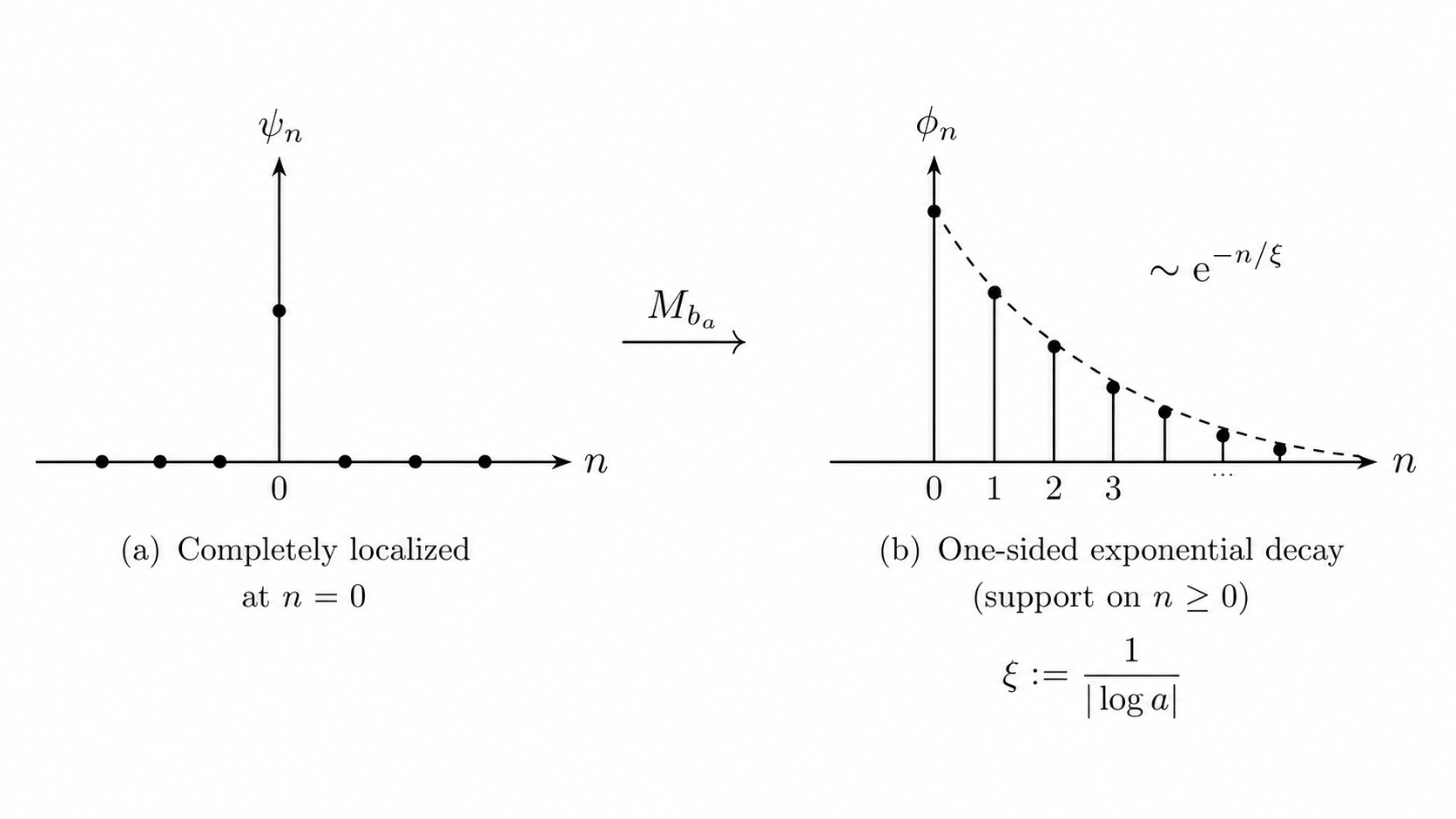}
    \caption{The soft shift $M_{b_a}$ acting on a fermion localized at
    $n=0$.  The resulting wave function is supported on $n\ge0$ and
    decays exponentially with localization length $\xi=1/|\log a|$.}
    \label{fig:onezero_soft_shift}
\end{figure}

Throughout, $u$ denotes an infinite Blaschke product with real zeros,
\begin{equation}\label{eq:standing}
    u(z)=\prod_{n=1}^{\infty}b_{a_n}(z),
    \qquad
    0\le a_n<1,
    \qquad
    \delta_n:=1-a_n,
    \qquad
    \sum_{n=1}^\infty\delta_n<\infty .
\end{equation}
The product converges uniformly on compact subsets of the unit disk
$\bbD$ and defines an inner function, that is, a bounded analytic function
on $\bbD$ with unimodular boundary values almost everywhere
(\cref{prop:model-spaces}).  The zero $a_n$ contributes the length scale
$\xi_{a_n}$, and $\xi_{a_n}\to\infty$; the zeros control how far a
localized particle is spread.  For example one may take
$a_n=1-2^{-n}$.

\subsection{A type \texorpdfstring{$\mathrm{II}_{\infty}$}{II-infinity}
half-chain algebra}\label{sec:typeIIoo}

Take two species, $\uparrow$ and $\downarrow$.  Start from the
translation-invariant pure product state in which, at every site, the
$\uparrow$-mode and the $\downarrow$-mode share one fermion in a Bell
pair; its covariance is $\frac12\left(\begin{smallmatrix}\Id&\Id\\
\Id&\Id\end{smallmatrix}\right)$.  Shifting all $\downarrow$-fermions by
$k$ sites makes $|k|$ Bell pairs cross a fixed cut, so the entanglement
across the cut is finite and the half-chain factor remains of type I.
Letting $k\to\infty$ does not help: the two flavours decorrelate and the
states converge in the weak-$*$ topology to the tracial state, which is
not pure.  Replacing the exact translation by the infinite sequence of
soft shifts $M_u$ keeps the global state pure and produces infinitely
many partially entangled modes across the cut (\cref{fig:soft_shift}).

\begin{figure}[htbp]
    \centering
    \includegraphics[width=\linewidth]{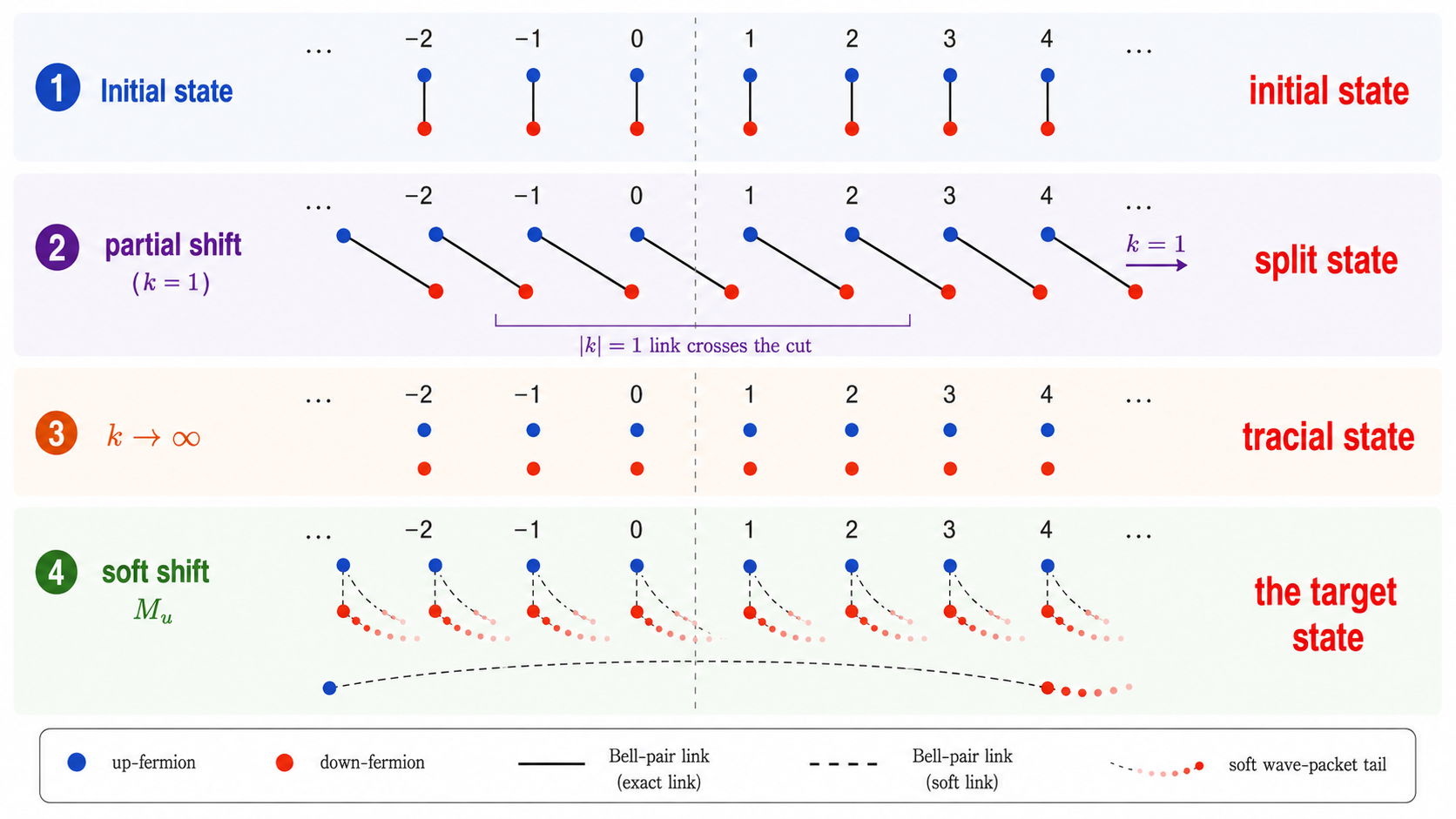}
    \caption{Exact partial shifts versus soft shifts.  Exact shifts give
    finite-range Bell-pair links; the soft shift $M_{u}$ produces long
    wave-packet tails and infinitely many entangled modes across the
    cut.}
    \label{fig:soft_shift}
\end{figure}

Translation invariance forces every entry of a covariance to be a
multiplication operator.  We choose
\begin{equation}\label{eq:cov_II_oo}
    C:=\frac{1}{2}
    \begin{pmatrix}
        \Id & M_u\\
        M_u^* & \Id
    \end{pmatrix}
    \quad\text{on } L^2(\bbT)\otimes\bbC^2 .
\end{equation}
Since $|u|=1$ almost everywhere, $M_u$ is unitary and $C^2=C$, so
$\omega_C$ is a translation-invariant pure quasifree state.  Its
right-half covariance is
\begin{equation}
    C_R=\frac{1}{2}
    \begin{pmatrix}
        \Id & T_u\\
        T_u^* & \Id
    \end{pmatrix}.
\end{equation}
As $u$ is inner, $T_u=M_u|_{H^2}$ is an isometry, and its failure to be
unitary is measured by the model space
\begin{equation}\label{eq:model-space}
    K_u:=H^2\ominus uH^2=\ker T_u^* .
\end{equation}
With respect to $H^2\otimes\bbC^2=(K_u\oplus0)\oplus(uH^2\oplus H^2)$ the
operator $C_R$ is block diagonal: $C_R=\frac12\Id$ on $K_u\oplus0$, while
on $uH^2\oplus H^2$, in the coordinates $(ug,h)\leftrightarrow(g,h)$, it is
the projection $\frac12\left(\begin{smallmatrix}1&1\\1&1\end{smallmatrix}
\right)$.  All the mixedness of the right-half state comes from the model
space.  By \cref{prop:model-spaces}, $\dim K_u$ equals the number of zeros
of $u$, so $\dim K_u=\infty$.

\begin{proposition}\label{prop:IIoo}
For $u$ as in \cref{eq:standing}, the half-chain algebra of the state
\cref{eq:cov_II_oo} is the hyperfinite factor of type
$\mathrm{II}_\infty$.
\end{proposition}

\begin{proof}
By \cref{lem:halfchain} it suffices to identify the factor of
$\omega_{C_R}$.  The covariance is a projection on the infinite-dimensional
space $uH^2\oplus H^2$ and equals $\frac12\Id$ on the infinite-dimensional
space $K_u$, so \cref{cor:IIoo-sector} applies.
\end{proof}

\subsection{Mixing ratios and type
\texorpdfstring{$\mathrm{III}_\lambda$}{III-lambda}}\label{sec:others}

For type $\mathrm{III}_\lambda$ the Bell pairs are biased.  For a mixing
ratio $0\le\kappa\le1$ put
\begin{equation}\label{eq:C_kappa}
C_{\kappa}
:=
\frac{1}{1+\kappa}
\begin{pmatrix}
\Id & \sqrt{\kappa}\,M_u\\
\sqrt{\kappa}\,M_u^* & \kappa\Id
\end{pmatrix},
\qquad
C_{\kappa,R}
=
\frac{1}{1+\kappa}
\begin{pmatrix}
\Id & \sqrt{\kappa}\,T_u\\
\sqrt{\kappa}\,T_u^* & \kappa\Id
\end{pmatrix}.
\end{equation}
Again $C_\kappa^2=C_\kappa$, so $C_\kappa$ defines a pure
translation-invariant quasifree state, and $\kappa=1$ is the state
\cref{eq:cov_II_oo}.  The decomposition $H^2=K_u\oplus uH^2$ shows that
$C_{\kappa,R}$ is a projection on $uH^2\oplus H^2$ and equals
$(1+\kappa)^{-1}\Id$ on $K_u\oplus0$.  Every vector of $K_u$ therefore
carries one mixed mode with density matrix
$(1+\kappa)^{-1}\diag(1,\kappa)$, up to the order of the basis states,
whose eigenvalue ratio is $\kappa$.

For an even number $\nu=2m$ of species and
$\boldsymbol\kappa=(\kappa_1,\ldots,\kappa_m)\in[0,1]^m$ put
\begin{equation}\label{eq:2m_flavor}
C_{\boldsymbol\kappa}
:=C_{\kappa_1}\oplus\cdots\oplus C_{\kappa_m}
\quad\text{on }L^2(\bbT)\otimes\bbC^{2m},
\end{equation}
and let $\omega_{\boldsymbol\kappa}$ be the corresponding pure
translation-invariant quasifree state.  Each summand, called a block, uses
the same inner function $u$; only the mixing ratio changes.  A block with
$\kappa_j=0$ is the product state in which one flavour is filled and the
other empty, so states of fewer species are included.

\begin{proposition}\label{prop:modular-family-types}
Let $\boldsymbol\kappa\in[0,1]^m$, put
$J_{\boldsymbol\kappa}:=\{j:0<\kappa_j<1\}$ and, if
$J_{\boldsymbol\kappa}\ne\varnothing$, let $G_{\boldsymbol\kappa}$ be the
closed subgroup of the multiplicative group $\bbR_+^\times$ generated by
$\{\kappa_j:j\in J_{\boldsymbol\kappa}\}$.  The half-chain algebra of
$\omega_{\boldsymbol\kappa}$ is a hyperfinite factor of the following type:
\begin{center}
\begin{tabular}{ll}
\toprule
mixing ratios & type\\
\midrule
$\kappa_j=0$ for all $j$ & $\mathrm{I}$\\
$J_{\boldsymbol\kappa}=\varnothing$ and $\kappa_j=1$ for some $j$
& $\mathrm{II}_\infty$\\
$G_{\boldsymbol\kappa}=\{\lambda^n:n\in\z\}$ with $0<\lambda<1$
& $\mathrm{III}_\lambda$\\
$G_{\boldsymbol\kappa}=\bbR_+^\times$ & $\mathrm{III}_1$\\
\bottomrule
\end{tabular}
\end{center}
These cases are exhaustive.
\end{proposition}

\begin{proof}
By the discussion above, the right-half covariance is a projection on an
infinite-dimensional subspace, direct sum of the operators
$(1+\kappa_j)^{-1}\Id$ on $m$ copies of $K_u$.  Copies with $\kappa_j=0$
belong to the projection part.  \Cref{prop:constant-sectors}, with
$t_j=(1+\kappa_j)^{-1}$ and eigenvalue ratios
$\min\{t_j/(1-t_j),(1-t_j)/t_j\}=\kappa_j$, gives the type.  A closed
subgroup of $\bbR_+^\times$ different from $\{1\}$ is either of the form
$\lambda^\z$ or all of $\bbR_+^\times$, so the cases are exhaustive.
\end{proof}

\begin{corollary}\label{cor:modular-family-examples}
Two species give type I for $\boldsymbol\kappa=(0)$, type
$\mathrm{II}_\infty$ for $\boldsymbol\kappa=(1)$ and type
$\mathrm{III}_\lambda$ for $\boldsymbol\kappa=(\lambda)$, $0<\lambda<1$.
Four species give type $\mathrm{III}_1$, for instance for
$\boldsymbol\kappa=(e^{-1},e^{-\sqrt2})$, since $\z+\sqrt2\,\z$ is dense
in $\bbR$.
\end{corollary}

Type $\mathrm{III}_1$ is already realized by one species
(\cref{prop:one-species}), but the four-species states of
\cref{cor:modular-family-examples} are those whose entropy can grow
arbitrarily slowly (\cref{sec:entropy}).

\section{Parent Hamiltonians}\label{sec:parent}

\subsection{Quadratic Hamiltonians and their ground states}
\label{sec:ground-states}

Let $h$ be a bounded self-adjoint operator on $\hilb=\ell^2(\z)\otimes
\bbC^\nu$.  It generates the quasifree dynamics
$\tau_t(c^*(f))=c^*(e^{ith}f)$ on $\CAR(\hilb)$ \cite{bratteli2013operator2},
with generator $\delta_h$ satisfying
\begin{equation}\label{eq:delta-h}
    \delta_h\bigl(c^*(f)\bigr)=i\,c^*(hf),
    \qquad
    \delta_h\bigl(c(f)\bigr)=-i\,c(hf).
\end{equation}
The $*$-algebra generated by all $c(f)$ is $\tau$-invariant and consists
of entire analytic elements, so it is a core for the generator
\cite[Cor.~3.1.7]{bratteli2013operator1}.  Since
$f\mapsto\delta_h(c(f))$ is norm continuous, the $*$-algebra
$\A_{\rm loc}$ of local elements, the polynomials in finitely many
$c_{x,a}$, is a core as well.  If $h=M_{h(z)}$ is translation invariant
with Bloch Hamiltonian $h(z)=\sum_rh_rz^r$, then
$h=\sum_{x,y}c_x^*h_{x-y}c_y$ in the usual notation
$c_x=(c_{x,1},\dots,c_{x,\nu})$, and for $\sum_r\|h_r\|<\infty$ the
generator acts on local elements as the convergent sum
$\delta_h(A)=i\sum_{x,y}[c_x^*h_{x-y}c_y,A]$.

A state $\omega$ is a ground state if
$-i\,\omega(A^*\delta_h(A))\ge0$ for every $A\in\A_{\rm loc}$, or
equivalently for every $A$ in the domain of the generator
\cite[Def.~5.3.18, Prop.~5.3.19]{bratteli2013operator2}.  The Fermi energy
is $0$, and $h$ is gapless if $0\in\Sp(h)$.

\begin{proposition}\label{prop:ground-states}
Let $h$ be bounded and self-adjoint with $\ker h=\{0\}$, and
$P:=\chi_{(-\infty,0)}(h)$.
\begin{enumerate}[label=\textup{(\roman*)}]
    \item $\omega_P$ is the unique ground state of the quasifree dynamics
    generated by $h$.
    \item If moreover the spectral projections $\chi_{(0,\epsilon)}(h)$
    and $\chi_{(-\epsilon,0)}(h)$ are nonzero for every $\epsilon>0$, then
    $\omega_P$ is the only even state $\omega$ with
    $-i\,\omega(A^*\delta_h(A))\ge0$ for all even $A\in\A_{\rm loc}$.
\end{enumerate}
\end{proposition}

The proof is given in \cref{app:groundstates}.  For translation-invariant
$h$, $\ker h=\{0\}$ means that $\{z:\det h(z)=0\}$ is a null set; the
hypothesis of (ii) holds when the bands take values in $(0,\epsilon)$ and
in $(-\epsilon,0)$ on sets of positive measure for every $\epsilon>0$.
Statement (ii) is used in \cref{sec:JWA} for the spin chain, where only
even fermionic observables have a local meaning.

Different blocks of \cref{eq:2m_flavor} are uncorrelated, so the parent
Hamiltonians below are direct sums over blocks, and it suffices to write
them for one block.

\subsection{Gapped parent Hamiltonians are long-ranged}\label{sec:gapped}

\begin{proposition}\label{prop:gapped-long-range}
Let $C$ be a translation-invariant projection whose half-chain factor is
not of type I, and let $h=M_{h(z)}$ be bounded with $0\notin\Sp(h)$ and
$\chi_{(-\infty,0)}(h)=C$.  Then $\sum_{r\in\z}|r|\,\|h_r\|_1=\infty$.
\end{proposition}

\begin{proof}
Suppose $\sum_r|r|\|h_r\|_1<\infty$.  As in \cref{lem:toeplitz-hankel},
$\|(\Id-\Pi)h\Pi\|_1\le\sum_{r\ge1}r\|h_{-r}\|_1$ and likewise for
$\Pi h(\Id-\Pi)$, so $[h,\Pi]\in\cS_1$.  Let $\Gamma$ be a contour in the
resolvent set enclosing $\Sp(h)\cap(-\infty,0)$ and crossing the real
axis at $0$.  Then
\begin{equation}
    [C,\Pi]=\frac{1}{2\pi i}\oint_\Gamma
    (\zeta-h)^{-1}[h,\Pi](\zeta-h)^{-1}\,\rmd\zeta\in\cS_1,
\end{equation}
hence $(\Id-\Pi)C\Pi=(\Id-\Pi)[C,\Pi]\Pi\in\cS_1\subset\cS_2$, and the
half-chain factor is of type I by \cref{lem:halfchain}.
\end{proof}

For the states of \cref{sec:constructing} the natural gapped parent is the
flat-band Hamiltonian
\begin{equation}\label{eq:flat_parent}
\Id-2C_{\kappa}
=
\begin{pmatrix}
    m_\kappa\Id & -t_\kappa M_u\\
    -t_\kappa M_u^* & -m_\kappa\Id
\end{pmatrix},
\qquad
m_{\kappa}:=\frac{\kappa-1}{1+\kappa},
\quad
t_{\kappa}:=\frac{2\sqrt{\kappa}}{1+\kappa}.
\end{equation}
Its spectrum is $\{\pm1\}$, so $\omega_\kappa$ is its unique ground state
by \cref{prop:ground-states}.  Writing $u(z)=\sum_{n\ge0}u_nz^n$ and
labelling the two flavours by $1$ and $2$, the second-quantized form is
\begin{equation}\label{eq:long_range_Hamiltonian}
H^{\rm flat}_{\kappa}
=
m_{\kappa}\sum_{x\in\z}
\bigl(c^*_{x,1}c_{x,1}-c^*_{x,2}c_{x,2}\bigr)
-
t_{\kappa}
\sum_{x\in\z}\sum_{n\geq0}
\bigl(u_n c^*_{x+n,1}c_{x,2}+\overline{u_n}\,c^*_{x,2}c_{x+n,1}\bigr),
\end{equation}
so $u_n$ is the inter-flavour hopping amplitude at distance $n$.  By
\cref{prop:gapped-long-range} every gapped parent of $\omega_\kappa$ with
$\kappa\ne0$ has $\sum_r|r|\|h_r\|=\infty$; for
\cref{eq:long_range_Hamiltonian} even $\sum_n|u_n|=\infty$
(\cref{lem:blaschke_long_range}), so the hopping is not absolutely
summable.

\subsection{Almost-local gapless parent Hamiltonians}\label{sec:almost-local}

The long range of $u_n$ comes from the singularity of $u$ at $z=1$, where
the zeros accumulate.  We damp it with the smooth function
\begin{equation}\label{eq:q-def}
    q(z):=\exp\bigl(-|1-z|^{-2}\bigr)\quad(z\ne1),
    \qquad q(1):=0 ,
\end{equation}
which is positive on $\bbT\setminus\{1\}$ and vanishes to infinite order
at $z=1$.

\begin{lemma}\label{lem:mollifier}
Let $F:\bbT\setminus\{1\}\to M_n$ be smooth with
$\|\partial_\theta^rF(e^{i\theta})\|\le C_r|\theta|^{-M_r}$ for
$0<|\theta|\le\pi$ and every $r\ge0$.  Then $qF$ extends to a smooth
function on $\bbT$ that vanishes to infinite order at $1$, and its Fourier
coefficients decay faster than any power.
\end{lemma}

\begin{proof}
Write $q(e^{i\theta})=e^{-g(\theta)}$ with
$g(\theta)=(4\sin^2(\theta/2))^{-1}\ge\theta^{-2}$.  The derivatives of $g$
obey $|g^{(j)}(\theta)|\le C_j|\theta|^{-2-j}$, so by the Faà di Bruno
formula \cite[p.~139]{Comtet1974}
$|\partial_\theta^jq(e^{i\theta})|\le C_j'|\theta|^{-3j}e^{-1/\theta^2}$.
By the Leibniz rule every derivative of $qF$ tends to $0$ as
$\theta\to0$, so $qF$ extends smoothly by $0$.  Smooth functions on $\bbT$
have rapidly decreasing Fourier coefficients \cite{Katznelson2004}.
\end{proof}

\begin{theorem}\label{thm:almost-local}
Let $u$ be as in \cref{eq:standing} and
$\boldsymbol\kappa\in[0,1]^m$.  The Bloch Hamiltonian
\begin{equation}\label{eq:almost_local_parent}
    h_{\boldsymbol\kappa}(z):=q(z)\bigl(\Id-2C_{\boldsymbol\kappa}(z)\bigr)
\end{equation}
is smooth on $\bbT$, so its hopping amplitudes satisfy
$\|(\hat h_{\boldsymbol\kappa})_r\|\le C_N(1+|r|)^{-N}$ for every $N$.  Its
bands are $\pm q(z)$, each of multiplicity $m$, it is gapless, and
$\omega_{\boldsymbol\kappa}$ is its unique ground state; the hypotheses of
\cref{prop:ground-states}(ii) are satisfied.
\end{theorem}

\begin{proof}
The entries of $C_{\boldsymbol\kappa}(z)$ are constants and constant
multiples of $u$ and $\bar u$.  By \cref{lem:blaschke-derivatives},
$|\partial_\theta^ru(e^{i\theta})|\le C_r(1+\sum_n\delta_n)^r|\theta|^{-2r}$,
so \cref{lem:mollifier} shows that $h_{\boldsymbol\kappa}$ is smooth.  Each
block $C_{\kappa_j}(z)$ is a projection of rank one, so
$\Id-2C_{\boldsymbol\kappa}(z)$ has eigenvalues $\pm1$, each of
multiplicity $m$, and $h_{\boldsymbol\kappa}(z)$ has eigenvalues
$\pm q(z)$.  Since $q>0$ off the point $1$, the negative spectral
projection of $h_{\boldsymbol\kappa}$ is $C_{\boldsymbol\kappa}$ and
$\ker h_{\boldsymbol\kappa}=\{0\}$.  Since $q$ is continuous, positive and
vanishes at $1$, it takes values in $(0,\epsilon)$ on a set of positive
measure for every $\epsilon>0$, so $\Sp(h_{\boldsymbol\kappa})$ accumulates
at $0$ from both sides.  \Cref{prop:ground-states} concludes the proof.
\end{proof}

The only Fermi point is $z=1$, where the bands behave like
$\pm e^{-1/\theta^2}$: the dispersion is flatter than any power.  For one
block, with $m_\kappa,t_\kappa$ as in \cref{eq:flat_parent} and
$q(z)=\sum_nq_nz^n$, $q(z)u(z)=\sum_nv_nz^n$, the real-space Hamiltonian
is
\begin{equation}\label{eq:almost_local_Hamiltonian}
\begin{split}
H_{\kappa}
={}&
m_{\kappa}\sum_{x,n\in\z}q_n
\bigl(c^*_{x+n,1}c_{x,1}-c^*_{x+n,2}c_{x,2}\bigr)\\
&-t_{\kappa}\sum_{x,n\in\z}
\bigl(v_n c^*_{x+n,1}c_{x,2}+\overline{v_n}\,c^*_{x,2}c_{x+n,1}\bigr),
\end{split}
\end{equation}
with $|q_n|+|v_n|\le C_N(1+|n|)^{-N}$ for every $N$.

\subsection{Exponentially local Hamiltonians}\label{sec:exponential}

If $\|h_r\|\le Ce^{-\mu|r|}$, then $h(z)$ is holomorphic in the annulus
$e^{-\mu}<|z|<e^{\mu}$; conversely a holomorphic extension to an annulus
gives exponential decay after shrinking the annulus.  Exponential locality
therefore means that $k\mapsto h(e^{ik})$ is real analytic, and analytic
band theory applies.

\begin{lemma}[Rellich]\label{lem:rellich}
Let $k\mapsto A(k)\in M_\nu(\bbC)$ be a real-analytic family of Hermitian
matrices on an open interval.  Then there are real-analytic functions
$\varepsilon_j(k)$ and an orthonormal basis of real-analytic vectors
$\varphi_j(k)$ with $A(k)\varphi_j(k)=\varepsilon_j(k)\varphi_j(k)$
\cite{Rellich1969}, \cite[Thm.~II.6.1]{Kato1995}.
\end{lemma}

\begin{lemma}\label{lem:hilbert-matrix}
For $\ell\ge1$ the $\ell\times\infty$ matrix
$G=[(x+j)^{-1}]_{0\le x<\ell,\,j\ge1}$ satisfies $\|G\|_1\le\log\ell+3$.
\end{lemma}

\begin{proof}
With $a_t=(t^x)_{0\le x<\ell}$ and $b_t=(t^{j-1})_{j\ge1}$ we have
$G=\int_0^1a_tb_t^*\,\rmd t$, hence
\begin{equation}
\|G\|_1\le\int_0^1\|a_t\|\,\|b_t\|\,\rmd t
=\int_0^1\frac{\sqrt{1-t^{2\ell}}}{1-t^2}\,\rmd t .
\end{equation}
On $[0,1-1/\ell]$ the integrand is at most $(1-t)^{-1}$, which contributes
at most $\log\ell$.  On $[1-1/\ell,1]$ we use $1-t^{2\ell}\le2\ell(1-t)$ and
$1-t^2\ge1-t$, which bound the remaining contribution by $2\sqrt2$.
\end{proof}

\begin{theorem}\label{thm:analytic-dichotomy}
Let $\nu<\infty$, let $h(z)=\sum_rh_rz^r\in M_\nu(\bbC)$ be Hermitian on
$\bbT$ with $\|h_r\|\le Ce^{-\mu|r|}$ for some $\mu>0$, and assume
$\det h(e^{ik})\not\equiv0$.  The set $Z$ of Fermi points, the zeros of
$\det h(e^{ik})$, is finite; near each $k_0\in Z$ the bands can be chosen
real analytic, and every band $\varepsilon_j$ with $\varepsilon_j(k_0)=0$
vanishes there to a finite order $m_j\ge1$.  The quasifree state
$\omega_{\rm F}$ with Fermi projection $P_{\rm F}=\chi_{(-\infty,0)}(h)$ is
the unique ground state, and:
\begin{enumerate}[label=\textup{(\alph*)}]
    \item if some $m_j$ is odd, $\M_R(\omega_{\rm F})$ is of type
    $\mathrm{III}_1$, and with $N_{\rm odd}\ge1$ the number of pairs
    $(k_0,j)$ with $m_j$ odd,
    \begin{equation}\label{eq:analytic-log}
        \frac{2\log2}{\pi^2}\,N_{\rm odd}\log\ell-C
        \le S_{\omega_{\rm F}}(\ell)\le C'\log\ell
        \qquad(\ell\ge2);
    \end{equation}
    \item if all $m_j$ are even, $\M_R(\omega_{\rm F})$ is of type I and
    $\sup_\ell S_{\omega_{\rm F}}(\ell)<\infty$.
\end{enumerate}
If $Z=\varnothing$ we are in case \textup{(b)}.  Suppose moreover that
$hw=\sum_Xa^Xw$ for every finitely supported $w\in\hilb$, where $X$ runs
over finite subsets of $\z$, $a^X$ is self-adjoint on
$\ell^2(X)\otimes\bbC^\nu$ and $\sup_{x}\sum_{X\ni x}\|a^X\|<\infty$, and
that the local terms $\Phi(X)=\sum_{x,y\in X}c_x^*a^X_{xy}c_y+b_X\Id$ satisfy
$\pi_{\rm F}(\Phi(X))\Omega_{\rm F}=0$ in the GNS representation of
$\omega_{\rm F}$.  This holds in particular if the Hamiltonian is frustration
free.  Then case \textup{(b)} holds.
\end{theorem}

\begin{proof}
\emph{Band structure.}  The function $\det h(e^{ik})$ is real analytic and
not identically zero, so $Z$ is finite, $\ker h=\{0\}$, and $\omega_{\rm F}$
is the unique ground state by \cref{prop:ground-states}.  Away from $Z$,
$P_{\rm F}$ is real analytic, being a Riesz projection.  Near $k_0\in Z$,
\cref{lem:rellich} gives analytic bands $\varepsilon_j$ with rank-one
spectral projections $P_j$; since $\prod_j\varepsilon_j=\det h\not\equiv0$,
no band vanishes identically, so each band vanishing at $k_0$ does so to a
finite order.  On a small interval to the right of $k_0$, respectively to
the left, the set $J_\pm$ of negative bands is constant, and $P_{\rm F}$
coincides there with the analytic projection
$P^{(\pm)}:=\sum_{j\in J_\pm}P_j$.  A zero of even order preserves the sign
of a band and one of odd order reverses it, so
$J_+\triangle J_-=\{j:\varepsilon_j(k_0)=0,\ m_j\ \text{odd}\}$, and the
jump
\begin{equation}
    J_{k_0}:=P^{(+)}(k_0)-P^{(-)}(k_0)
    =\sum_{j\in J_+\setminus J_-}P_j(k_0)-\sum_{j\in J_-\setminus J_+}P_j(k_0)
\end{equation}
satisfies $\|J_{k_0}\|_2^2=|J_+\triangle J_-|$, because the $P_j(k_0)$ are
mutually orthogonal rank-one projections.  Let $Z_{\rm odd}=\{k_s\}$ be the
set of points with $J_{k_s}\ne0$, so that
$N_{\rm odd}=\sum_s\|J_{k_s}\|_2^2$.  After redefinition on the finite set
$Z\setminus Z_{\rm odd}$, the function $P_{\rm F}$ is continuous off
$Z_{\rm odd}$, and on each closed arc between consecutive points of
$Z_{\rm odd}$ it is the restriction of a function analytic on a
neighbourhood of the arc.

\emph{Case (b).}  Here $Z_{\rm odd}=\varnothing$ and $P_{\rm F}$ is real
analytic on $\bbT$, so $\|\hat P_{{\rm F},r}\|\le Ce^{-\mu'|r|}$ for some
$\mu'>0$.  By \cref{lem:toeplitz-hankel}(i) and \cref{lem:halfchain} the
half-chain factor is of type I.  By \cref{lem:entropy-bounds} and
\cref{lem:toeplitz-hankel}(ii),
$S_{\omega_{\rm F}}(\ell)\le2\log2\sum_{r\ne0}|r|\,\|\hat P_{{\rm F},r}\|_1
<\infty$ uniformly in $\ell$.

\emph{Case (a), type.}  Now $P_{\rm F}$ is a piecewise continuous
projection-valued symbol with at least one discontinuity.  The essential
spectrum of the block Toeplitz operator $T_{P_{\rm F}}$ then contains a
nondegenerate interval \cite[Sec.~6.2]{vanLuijk2025}, and the factor is of
type $\mathrm{III}_1$ by \cite[Lemma~26]{vanLuijk2025} and
\cref{lem:halfchain}.

\emph{Case (a), entropy.}  For $k_0\in\bbT$ let $\sigma_{k_0}$ be the
sawtooth function with $\sigma_{k_0}(e^{ik})=(\pi-t)/(2\pi)$ for
$k-k_0\equiv t\in(0,2\pi)$; it jumps by $+1$ at $k_0$ and
$(\hat\sigma_{k_0})_r=e^{-irk_0}/(2\pi ir)$ for $r\ne0$.  Then
\begin{equation}
    A:=P_{\rm F}-\sum_sJ_{k_s}\sigma_{k_s}
\end{equation}
is continuous on $\bbT$ and smooth up to the endpoints on each arc between
consecutive points of $Z_{\rm odd}$.  Two integrations by parts give
$\|\hat A_r\|_1\le C|r|^{-2}$.  Let $X_\ell=(\Id-Q_\ell)P_{\rm F}Q_\ell$.

For the upper bound, \cref{lem:toeplitz-hankel}(ii) gives
$\|(\Id-Q_\ell)M_AQ_\ell\|_1\le C\sum_{r\ge1}\min\{r,\ell\}r^{-2}
\le C(1+\log\ell)$.  The operator $(\Id-Q_\ell)M_{\sigma_{k_0}}Q_\ell$ has
the kernel $e^{-ik_0(y-x)}/(2\pi i(y-x))$ for $x\in I_\ell$, $y\notin I_\ell$.
The phases are diagonal unitaries, and the parts $y<0$ and $y\ge\ell$ are,
after the substitutions $y=-j$ and $y=\ell-1+j$, $x\mapsto\ell-1-x$, copies
of the matrix of \cref{lem:hilbert-matrix}.  Hence
$\|(\Id-Q_\ell)M_{\sigma_{k_0}}Q_\ell\|_1\le\pi^{-1}(\log\ell+3)$, so
$\|X_\ell\|_1\le C'(1+\log\ell)$, and \cref{lem:entropy-bounds} gives the
upper bound in \cref{eq:analytic-log}.

For the lower bound, \cref{lem:entropy-bounds} and
\cref{lem:toeplitz-hankel}(ii) give
$S_{\omega_{\rm F}}(\ell)\ge4\log2\sum_{1\le|r|\le\ell}|r|\,
\|\hat P_{{\rm F},r}\|_2^2$.  Write
$\hat P_{{\rm F},r}=j_r+\hat A_r$ with
$j_r=(2\pi ir)^{-1}\sum_se^{-irk_s}J_{k_s}$.  Since $\|j_r\|_2\le C/|r|$
and $\|\hat A_r\|_2\le C/r^2$, we have
$|r|\,\|\hat P_{{\rm F},r}\|_2^2\ge|r|\,\|j_r\|_2^2-C/r^2$, and
\begin{equation}
    |r|\,\|j_r\|_2^2=\frac{1}{4\pi^2|r|}
    \Bigl(N_{\rm odd}+\sum_{s\ne s'}e^{ir(k_s-k_{s'})}
    \Tr\bigl(J_{k_s}^*J_{k_{s'}}\bigr)\Bigr).
\end{equation}
Now $\sum_{1\le|r|\le\ell}|r|^{-1}\ge2\log\ell$, while for
$\theta\notin2\pi\z$ the partial sums of $\sum_{r\ge1}\cos(r\theta)/r$ are
bounded by $|\sin(\theta/2)|^{-1}$ by summation by parts.  The cross terms
therefore contribute $O(1)$, and the lower bound in
\cref{eq:analytic-log} follows.

\emph{Frustration-free case.}  In the state $\omega_{\rm F}$ the variance of
the quadratic operator $\Phi(X)$ equals $\|(\Id-P_{\rm F})a^XP_{\rm F}\|_2^2$,
as follows from Wick's rule \cref{eq:quasifree}.  Hence
$\pi_{\rm F}(\Phi(X))\Omega_{\rm F}=0$ forces $(\Id-P_{\rm F})a^XP_{\rm F}=0$,
that is $[a^X,P_{\rm F}]=0$, and $P_{\rm F}$ maps the finite-dimensional
range of $a^X$, which consists of finitely supported vectors, into
itself.  Thus for each $w\in\Ran a^X$ both $\hat w(k)$ and
$P_{\rm F}(k)\hat w(k)$ are trigonometric polynomials.  Suppose
$k_0\in Z_{\rm odd}$.  The analytic function $P_{\rm F}\hat w$ coincides
with $P^{(+)}\hat w$ to the right of $k_0$ and with $P^{(-)}\hat w$ to the
left, so $(P^{(+)}-P^{(-)})\hat w\equiv0$ near $k_0$ by analyticity, and
$P_j(k)\hat w(k)=0$ near $k_0$ for every $j\in J_+\triangle J_-$.  For
$v\in\bbC^\nu$ only the sets $X\ni0$ contribute to
$h(\delta_0\otimes v)=\sum_Xa^X(\delta_0\otimes v)$, and
$\sum_{X\ni0}\|a^X\|<\infty$, so the sum converges absolutely in $\ell^2$;
hence
$h(e^{ik})v=\sum_X\bigl(a^X(\delta_0\otimes v)\bigr)^\wedge(k)$ in $L^2$,
and every term is orthogonal to $\varphi_j(k)$ near $k_0$.  Hence
$\varepsilon_j(k)=\langle\varphi_j(k),h(e^{ik})\varphi_j(k)\rangle=0$ near
$k_0$, so $\varepsilon_j\equiv0$, contradicting $\det h\not\equiv0$.
Therefore $Z_{\rm odd}=\varnothing$.
\end{proof}

\begin{remark}
The hypothesis $\det h\not\equiv0$ excludes a flat band at the Fermi
energy.  If a band vanishes identically, the ground states of the
quasifree dynamics include $\omega_C$ for every translation-invariant
covariance $P_{\rm F}\le C\le P_{\rm F}+P_0$, with $P_0$ the projection
onto the flat band, and the ground state is not unique.  For instance,
every state $\omega_{\boldsymbol\kappa}$ of \cref{sec:others} is a ground
state of the zero Hamiltonian.
\end{remark}

\begin{remark}
An odd-order crossing produces a jump of the Fermi projection and type
$\mathrm{III}_1$, while an even-order contact, such as an analytic
quadratic band touching, produces no jump and gives type I together with
an area law, although the Hamiltonian is gapless.  The Hamiltonians of
\cref{thm:almost-local} escape \cref{thm:analytic-dichotomy} because $q$
is flat at $z=1$: the functions $qu$ and $q\bar u$ are smooth but not
analytic, so their Fourier coefficients decay faster than any power, yet
not exponentially.
\end{remark}

\section{Entropy scaling}\label{sec:entropy}

We now estimate the entropy of the states $\omega_{\boldsymbol\kappa}$ of
\cref{sec:others}.  Blocks are uncorrelated, so their entropies add, and
blocks with $\kappa_j=0$ are product states with zero entropy.  It
therefore suffices to treat one block $C_\kappa$ with $0<\kappa\le1$.
Throughout this section $u$ is as in \cref{eq:standing},
$I_\ell=\{0,\dots,\ell-1\}$, and $S_\kappa(\ell)$ is the entropy of
$\omega_\kappa$ on $I_\ell$.

\subsection{Finite-interval entropy and \texorpdfstring{$D_\ell$}{D}}
\label{sec:interval-entropy}

\begin{proposition}\label{prop:interval-entropy}
Let $0<\kappa\le1$, let $T$ be a contraction on $\bbC^n$ with singular
values $\sigma_1,\dots,\sigma_n$, and let
\begin{equation}\label{eq:interval-cov}
C=\frac{1}{1+\kappa}
\begin{pmatrix}
\Id & \sqrt{\kappa}\,T\\
\sqrt{\kappa}\,T^* & \kappa\Id
\end{pmatrix}
\quad\text{on }\bbC^n\oplus\bbC^n .
\end{equation}
Put $d_j:=1-\sigma_j^2$, $D:=\sum_jd_j=\Tr(\Id-T^*T)$ and
$\gamma_\kappa:=4\kappa/(1+\kappa)^2$.  The eigenvalues of $C$ are
$\frac12(1\pm\sqrt{1-\gamma_\kappa d_j})$, and
\begin{equation}\label{eq:entropy-F}
\Tr\hb(C)=\sum_{j=1}^{n}F_\kappa(d_j),
\qquad
F_\kappa(d):=2\,\hb\Bigl(\tfrac12\bigl(1-\sqrt{1-\gamma_\kappa d}\bigr)\Bigr).
\end{equation}
There is a universal constant $C_0$ such that
\begin{equation}\label{eq:F-bounds}
2\gamma_\kappa\log2\;d\le F_\kappa(d)
\le C_0\,\gamma_\kappa d\log\frac{e}{\gamma_\kappa d},
\qquad0<d\le1,
\end{equation}
and consequently
\begin{equation}\label{eq:entropy-DL}
2\gamma_\kappa\log2\;D\le\Tr\hb(C)
\le C_0\,\gamma_\kappa D\log\frac{en}{\gamma_\kappa D}
\quad(D>0).
\end{equation}
\end{proposition}

\begin{proof}
A singular-value decomposition of $T$ turns $C$ into a direct sum of the
$2\times2$ matrices
$(1+\kappa)^{-1}\left(\begin{smallmatrix}1&\sqrt\kappa\sigma_j\\
\sqrt\kappa\sigma_j&\kappa\end{smallmatrix}\right)$, with trace $1$ and
determinant $\gamma_\kappa d_j/4$; this gives the eigenvalues, and
\cref{eq:entropy-F} follows from $\hb(t)=\hb(1-t)$.  For the lower bound in
\cref{eq:F-bounds} apply \cref{eq:topsoe} with $t(1-t)=\gamma_\kappa d/4$.
For the upper bound put $x=\gamma_\kappa d\in(0,1]$ and use
$x/4\le\frac12(1-\sqrt{1-x})\le x/2$, the monotonicity of $\hb$ on
$[0,\frac12]$ and $\hb(t)\le t\log(e/t)$, which follows from
$(1-t)\log\frac1{1-t}\le t$.  Summing gives the lower bound in
\cref{eq:entropy-DL}; the upper bound follows from Jensen's inequality for
the concave function $x\mapsto x\log(e/x)$ applied to the $n$ numbers
$\gamma_\kappa d_j$.
\end{proof}

The covariance of $\omega_\kappa$ on $I_\ell$ is \cref{eq:interval-cov}
with $T=T_\ell:=Q_\ell T_uQ_\ell$, the lower triangular Toeplitz matrix
$[u_{j-k}]_{0\le k\le j<\ell}$, which is a contraction.  We write
\begin{equation}\label{eq:DL-definition}
    D_\ell:=\Tr(\Id-T_\ell^*T_\ell),
\end{equation}
so that $S_\kappa(\ell)=\sum_jF_\kappa(d_j(\ell))$ and
\cref{eq:entropy-DL} applies with $n=\ell$.

\begin{corollary}\label{cor:kappa-comparison}
For $\kappa,\kappa'\in(0,1]$ there are constants $c,C>0$, depending only
on $\kappa$ and $\kappa'$, with $cS_{\kappa'}(\ell)\le S_\kappa(\ell)\le
CS_{\kappa'}(\ell)$ for every $\ell$ and every choice of zeros.
\end{corollary}

\begin{proof}
On $(0,1]$ the ratio $F_\kappa(d)/F_{\kappa'}(d)$ is positive and
continuous and tends to $\gamma_\kappa/\gamma_{\kappa'}$ as $d\to0$, so it
is bounded above and below by positive constants.
\end{proof}

At fixed $\kappa$ the entropy is therefore governed by the singular values
of $T_\ell$, hence by the Blaschke zeros alone.

\begin{lemma}\label{lem:DL-formulas}
With $u(z)=\sum_{m\ge0}u_mz^m$,
\begin{equation}\label{eq:DL-coefficients}
D_\ell=\sum_{m\ge1}\min\{m,\ell\}\,|u_m|^2
=\Tr\bigl(Q_\ell P_{K_u}Q_\ell\bigr),
\end{equation}
where $P_{K_u}$ is the projection onto $K_u$.  In particular
$D_\ell\uparrow\dim K_u=\infty$, and $S_\kappa(\ell)\to\infty$.
\end{lemma}

\begin{proof}
$\Tr(T_\ell^*T_\ell)=\sum_{m=0}^{\ell-1}(\ell-m)|u_m|^2$ and
$\sum_m|u_m|^2=\|u\|_2^2=1$ give the first formula.  For the second, the
formula $T_u^*z^j=\sum_{m\le j}\overline{u_m}z^{j-m}$ shows
$Q_\ell T_u^*Q_\ell=T_u^*Q_\ell$, hence
$T_\ell T_\ell^*=Q_\ell T_uT_u^*Q_\ell=Q_\ell(\Id-P_{K_u})Q_\ell$, and
$T_\ell T_\ell^*$ and $T_\ell^*T_\ell$ have the same trace.  Monotone
convergence and \cref{eq:entropy-DL} give the last statement.
\end{proof}

\subsection{Estimating \texorpdfstring{$D_\ell$}{D} from the zeros}

\begin{proposition}\label{prop:radial}
Put $\epsilon_\ell:=1-e^{-1/(2\ell)}$, so that
$(4\ell)^{-1}\le\epsilon_\ell\le(2\ell)^{-1}$.  Then
\begin{equation}\label{eq:DL-zero-formula}
D_\ell
\asymp
\frac{1}{\epsilon_\ell}
\int_{-\pi}^{\pi}
\min
\Bigl\{
1,\;
\epsilon_\ell
\sum_{n\ge1}
\frac{\delta_n}{(\delta_n+\epsilon_\ell)^2+t^2}
\Bigr\}
\,\rmd t
\end{equation}
with universal implied constants.
\end{proposition}

\begin{proof}
Let $r_\ell=e^{-1/(2\ell)}$ and $\epsilon=\epsilon_\ell$.  Since
$\sum_m|u_m|^2=1$,
\begin{equation}
1-\frac{1}{2\pi}\int_{-\pi}^{\pi}|u(r_\ell e^{it})|^2\,\rmd t
=\sum_{m\ge1}\bigl(1-e^{-m/\ell}\bigr)|u_m|^2 ,
\end{equation}
and $(1-e^{-1})\min\{x,1\}\le1-e^{-x}\le\min\{x,1\}$ together with
\cref{eq:DL-coefficients} shows that $D_\ell$ is comparable to $\ell$
times the left-hand side.  For a single factor
$1-|b_a(z)|^2=(1-a^2)(1-|z|^2)/|1-az|^2$, and uniformly in
$t\in[-\pi,\pi]$ and $a\in[0,1)$ one has
$(1-a^2)(1-r_\ell^2)\asymp\delta\epsilon$ and
$|1-ar_\ell e^{it}|^2\asymp(\delta+\epsilon)^2+t^2$, with $\delta=1-a$.
Hence $x_n(t):=1-|b_{a_n}(r_\ell e^{it})|^2\asymp
\delta_n\epsilon/((\delta_n+\epsilon)^2+t^2)$, and
$|u(r_\ell e^{it})|^2=\prod_n(1-x_n(t))$.  Finally
$(1-e^{-1})\min\{1,\sum_nx_n\}\le1-\prod_n(1-x_n)\le\min\{1,\sum_nx_n\}$
for $0\le x_n\le1$.
\end{proof}

\begin{corollary}\label{cor:DL-upper}
$D_\ell\le C\sum_{n\ge1}\min\{1,\ell\delta_n\}$ with a universal constant
$C$.
\end{corollary}

\begin{proof}
Drop the minimum in \cref{eq:DL-zero-formula} and use
$\int_\bbR\delta\,\rmd t/((\delta+\epsilon)^2+t^2)=\pi\delta/(\delta+\epsilon)
\le\pi\min\{1,\delta/\epsilon\}$ and $\epsilon_\ell\ge(4\ell)^{-1}$.
\end{proof}

For upper bounds sharper than \cref{eq:entropy-DL} we truncate the
product.  The following continuity estimate combines Mirsky's
eigenvalue perturbation inequality \cite{Mirsky1960Symmetric,
Bhatia1997Matrix} with the two-level case of the Fannes--Audenaert
inequality \cite{Audenaert2007Sharp}, $|\hb(s)-\hb(t)|\le\hb(|s-t|)$, and
the concavity of $\hb$.

\begin{lemma}\label{lem:entropy-continuity}
Let $A,B$ be positive contractions on a space of dimension $k$ with
$\|A-B\|_1\le k/2$.  Then
$|\Tr\hb(A)-\Tr\hb(B)|\le k\,\hb\bigl(\|A-B\|_1/k\bigr)$.
\end{lemma}

\begin{proposition}\label{prop:truncated}
For $N\ge0$ let $u^{(N)}=\prod_{n=1}^Nb_{a_n}$, with $u^{(0)}=1$, and
$\tau_N:=\sum_{n>N}\delta_n$.  If $\sqrt{\gamma_\kappa\tau_N/2}\le\frac12$,
then
\begin{equation}\label{eq:truncated-bound}
S_\kappa(\ell)\le2N\log2+2\ell\,\hb\Bigl(\sqrt{\gamma_\kappa\tau_N/2}\Bigr).
\end{equation}
\end{proposition}

\begin{proof}
Let $C^{(N)}$ be the interval covariance built from $u^{(N)}$.  Since
$\dim K_{u^{(N)}}=N$, \cref{lem:DL-formulas} shows that at most $N$ of the
numbers $d_j$ belonging to $u^{(N)}$ are nonzero, and $F_\kappa\le2\log2$,
so $\Tr\hb(C^{(N)})\le2N\log2$.  Write $u=u^{(N)}u_{>N}$ with
$u_{>N}=\prod_{n>N}b_{a_n}$.  Since $|u^{(N)}|=1$ on $\bbT$ and
$u_{>N}(0)=\prod_{n>N}a_n$,
\begin{equation}
\|u-u^{(N)}\|_2^2=\|u_{>N}-1\|_2^2=2\Bigl(1-\prod_{n>N}a_n\Bigr)\le2\tau_N .
\end{equation}
For $v\in H^2$ we have $\|Q_\ell M_vQ_\ell\|_2^2\le\ell\|v\|_2^2$ and hence
$\|Q_\ell M_vQ_\ell\|_1\le\ell\|v\|_2$.  The two covariances differ only in
the off-diagonal blocks, so
$\|C-C^{(N)}\|_1\le\sqrt{\gamma_\kappa}\,\ell\sqrt{2\tau_N}$, and
\cref{lem:entropy-continuity} with $k=2\ell$ gives the claim.
\end{proof}

\subsection{Lacunary zeros}\label{sec:lacunary}

Call the zeros lacunary if $\delta_{n+1}\le\delta_n/2$ for all $n$, and
put
\begin{equation}
    \mathcal N(L):=\#\{n:\delta_n\ge1/L\}.
\end{equation}

\begin{proposition}\label{prop:lacunary}
For lacunary zeros and $\ell\ge2$,
\begin{equation}\label{eq:lacunary}
    c\,\gamma_\kappa\,\mathcal N(\ell)\le S_\kappa(\ell)
    \le2\log2\;\mathcal N(\ell^4)+3,
\end{equation}
with a universal constant $c>0$.
\end{proposition}

\begin{proof}
Lower bound.  Let $\delta_n\ge1/\ell$; then $\delta_n\ge\epsilon_\ell$.
On the arc $A_n=[\delta_n/4,\delta_n/2]$ we have
$(\delta_n+\epsilon_\ell)^2+t^2\le5\delta_n^2$, so the $n$-th summand in
\cref{eq:DL-zero-formula} gives
$\min\{1,\epsilon_\ell\sum_m\cdots\}\ge\epsilon_\ell/(5\delta_n)$ on
$A_n$.  The arcs $A_n$ are disjoint by lacunarity, and each contributes at
least $1/20$ to the integral divided by $\epsilon_\ell$.  Hence
$D_\ell\ge c'\mathcal N(\ell)$, and \cref{eq:entropy-DL} gives the lower
bound.

Upper bound.  Let $N=\mathcal N(\ell^4)$.  Then $\delta_{N+1}<\ell^{-4}$,
and by lacunarity $\tau_N\le2\delta_{N+1}<2\ell^{-4}$, so
$\sqrt{\gamma_\kappa\tau_N/2}<\ell^{-2}\le\frac14$.  By
\cref{prop:truncated} and $\hb(t)\le t\log(e/t)$,
$S_\kappa(\ell)\le2N\log2+2\ell^{-1}(1+2\log\ell)$, and the last term is
less than $3$.
\end{proof}

\begin{corollary}\label{cor:lacunary-examples}
For fixed $\kappa\in(0,1]$:
\begin{enumerate}[label=\textup{(\roman*)}]
    \item $\delta_n=2^{-n}$ gives $S_\kappa(\ell)\asymp\log\ell$;
    \item $\delta_n=2^{-n^2}$ gives $S_\kappa(\ell)\asymp\sqrt{\log\ell}$;
    \item $\delta_n=2^{-2^n}$ gives $S_\kappa(\ell)\asymp\log\log\ell$.
\end{enumerate}
\end{corollary}

\begin{proof}
All three sequences are lacunary.  In the three cases
$\mathcal N(L)=\lfloor\log_2L\rfloor$,
$\mathcal N(L)=\lfloor\sqrt{\log_2L}\rfloor$ and
$\mathcal N(L)=\lfloor\log_2\log_2L\rfloor$, respectively, and
$\mathcal N(\ell^4)$ is at most $4\mathcal N(\ell)+4$,
$2\mathcal N(\ell)+2$ and $\mathcal N(\ell)+3$.
\end{proof}

\subsection{Polynomial zeros}\label{sec:polynomial}

\begin{proposition}\label{prop:polynomial}
Let $\delta_n=n^{-\alpha}$ with $\alpha>1$.  Then
$D_\ell\asymp\ell^{1/(\alpha+1)}$ and
\begin{equation}\label{eq:entropy-polynomial-bounds}
c_\kappa\ell^{1/(\alpha+1)}\le S_\kappa(\ell)
\le C_\kappa\ell^{1/(\alpha+1)}\log\ell
\qquad(\ell\ge2).
\end{equation}
\end{proposition}

\begin{proof}
With $\rho=\epsilon+|t|$, split the sum in \cref{eq:DL-zero-formula} at
$n\asymp\rho^{-1/\alpha}$: for smaller $n$ the summand is of order
$n^{\alpha}$, for larger $n$ of order $n^{-\alpha}\rho^{-2}$, and both parts
are of order $\rho^{-1-1/\alpha}$.  Hence
$D_\ell\asymp\epsilon^{-1}\int_{-\pi}^{\pi}
\min\{1,\epsilon(\epsilon+|t|)^{-1-1/\alpha}\}\,\rmd t$.  The two entries of
the minimum are comparable at $t_\epsilon\asymp\epsilon^{\alpha/(\alpha+1)}$,
so $D_\ell\asymp t_\epsilon/\epsilon+\int_{t_\epsilon}^1
t^{-1-1/\alpha}\,\rmd t\asymp\epsilon^{-1/(\alpha+1)}$.  Now apply
\cref{eq:entropy-DL}.
\end{proof}

Removing the logarithmic factor in the upper bound would require finer
information on the singular values of $T_\ell$ than the trace $D_\ell$.

\subsection{Arbitrarily slow entropy growth}\label{sec:slow}

\begin{proposition}\label{prop:slow}
Let $m\ge1$ and let $f:\bbN\to(0,\infty)$ be nondecreasing with
$f(\ell)\to\infty$.  There are lacunary zeros $0<a_1<a_2<\cdots<1$ such
that, for every $\boldsymbol\kappa\in[0,1]^m$, the state
$\omega_{\boldsymbol\kappa}$ satisfies $S_{\omega_{\boldsymbol\kappa}}(\ell)
\le f(\ell)$ for every $\ell\ge1$.  If some $\kappa_j>0$, then moreover
$S_{\omega_{\boldsymbol\kappa}}(\ell)\to\infty$.
\end{proposition}

\begin{proof}
For $n\ge1$ let $L_n:=\min\{\ell\ge2:f(\ell)\ge m(2n\log2+3)\}$; this
sequence is nondecreasing and tends to infinity.  Choose $\eta\in(0,\frac14]$
so small that $2m\ell\,\hb(\sqrt{\eta/2})\le f(\ell)$ for the finitely many
$\ell<L_1$, and define $\delta_0:=\eta$ and
$\delta_n:=\min\{\delta_{n-1}/2,(2L_n)^{-4}\}$.  The sequence is lacunary,
$\sum_n\delta_n\le\eta$, and $\delta_n<L_n^{-4}$.

For $\ell<L_1$, \cref{prop:truncated} with $N=0$ gives
$S_\kappa(\ell)\le2\ell\,\hb(\sqrt{\eta/2})$ for each block, using
$\gamma_\kappa\le1$.  For $\ell\ge L_1$ let $n$ be such that
$L_n\le\ell<L_{n+1}$.  If $\delta_k\ge\ell^{-4}$, then $L_k<\ell$, hence
$k\le n$; so $\mathcal N(\ell^4)\le n$, and \cref{prop:lacunary} gives
$S_\kappa(\ell)\le2n\log2+3$ for each block.  Summing over the at most
$m$ blocks with $\kappa_j>0$ and using that $f$ is nondecreasing gives
$S_{\omega_{\boldsymbol\kappa}}(\ell)\le f(\ell)$ in both cases.
Divergence follows from \cref{lem:DL-formulas}.
\end{proof}

Combining \cref{cor:lacunary-examples,prop:polynomial,prop:slow} with
\cref{prop:modular-family-types} gives the fermionic version of
\cref{thm:main}(iii).  For instance, the dyadic zeros give states of type
$\mathrm{II}_\infty$, $\mathrm{III}_\lambda$ and $\mathrm{III}_1$ with the
same growth $\log\ell$, while the doubly exponential zeros give type
$\mathrm{III}_1$ states whose entropy grows like $\log\log\ell$.

\section{Jordan--Wigner--Araki transformation}\label{sec:JWA}

It remains to transport the fermionic states to a spin chain.  We use the
one-dimensional Jordan--Wigner transformation in the form introduced by
Araki for the two-sided infinite chain
\cite{Araki1984,ArakiMatsui1985,KEYL_2006,Aschbacher_2023}.  Let
$\A_{\rm f}=\CAR(\ell^2(\z)\otimes\bbC^\nu)$ with grading $\Theta$, and let
$\A_{\rm s}=\bigotimes_{x\in\z}M_{2^\nu}$, where the on-site algebra is
identified with $\CAR(\bbC^\nu)$ and graded by the on-site fermion parity
$\Gamma_x$; the grading $\Theta_{\rm s}$ of $\A_{\rm s}$ is the product of
the automorphisms $\mathrm{Ad}\,\Gamma_x$.  We write
$E_{\rm s}=\frac12(\mathrm{id}+\Theta_{\rm s})$ and
$E_{\rm f}=\frac12(\mathrm{id}+\Theta)$ for the conditional expectations
onto the even parts $\A_{\rm s}^+$ and $\A_{\rm f}^+$.

\subsection{Local structure and the dual state}

\begin{lemma}\label{lem:jw-local}
For every finite interval $I$ there is a grading-preserving
$*$-isomorphism $\vartheta_I:\A_{\rm f}(I)\to\A_{\rm s}(I)$, and
$\vartheta_{I'}(A)=\vartheta_I(A)$ for $I\subset I'$ and even
$A\in\A_{\rm f}(I)$.  Consequently the $\vartheta_I$ define a
$*$-isomorphism $\vartheta:\A_{\rm f}^+\to\A_{\rm s}^+$ which commutes with
lattice translations and maps even elements supported in an interval to
elements supported in the same interval.
\end{lemma}

\begin{proof}
On $I=[a,b]$ order the modes lexicographically and put
$\vartheta_I(c_{x,j})=\Gamma_a\cdots\Gamma_{x-1}\,\gamma_x(c_j)$, where
$\gamma_x$ is the identification of $\CAR(\bbC^\nu)$ with the on-site
algebra at $x$.  The strings of two odd operators combine to strings
inside $I$, which gives the compatibility on even elements; see
\cite{Araki1984} and \cite[Props.~4.1 and~4.3]{MoonNachtergaele2018}.  The
map on $\A_{\rm f}^+$ is obtained by norm closure.
\end{proof}

\begin{definition}\label{def:jw-state}
The spin state dual to an even state $\omega_{\rm f}$ on $\A_{\rm f}$ is
$\omega_{\rm s}:=\omega_{\rm f}\circ\vartheta^{-1}\circ E_{\rm s}$.
\end{definition}

\begin{lemma}\label{lem:jw-entropy}
If $\omega_{\rm f}$ is even and translation invariant, then
$\omega_{\rm s}$ is even and translation invariant, and for every finite
interval $I$ the density matrices of $\omega_{\rm f}|_{\A_{\rm f}(I)}$ and
$\omega_{\rm s}|_{\A_{\rm s}(I)}$ are unitarily equivalent.  In particular
$S_{\omega_{\rm s}}(I)=S_{\omega_{\rm f}}(I)$.
\end{lemma}

\begin{proof}
Translation invariance follows from \cref{lem:jw-local}.  Since
$\vartheta_I$ preserves the grading and both states vanish on odd
elements, $\omega_{\rm s}\circ\vartheta_I=\omega_{\rm f}$ on all of
$\A_{\rm f}(I)$.  Both local algebras are full matrix algebras, so
$\vartheta_I$ is implemented by a unitary.
\end{proof}

Let $\M_{R/L}^{\rm f}$ and $\M_{R/L}^{\rm s}$ be the half-chain algebras of
$\omega_{\rm f}$ and $\omega_{\rm s}$ in their GNS representations.

\begin{lemma}\label{lem:jw-halfchain}
Let $\omega_{\rm f}$ be a translation-invariant pure gauge-invariant
quasifree state whose half-chain algebra is not of type I.  Then
$\omega_{\rm s}$ is pure, $\M_R^{\rm s}\cong\M_R^{\rm f}$, and Haag duality
$\M_L^{\rm s}=(\M_R^{\rm s})'$ holds.
\end{lemma}

\begin{proof}
Since $\M_R^{\rm f}$ is not of type I, the parity of the left half-chain
is not unitarily implementable in the GNS representation of
$\omega_{\rm f}$ \cite[Lemma~18 and Prop.~16]{vanLuijk2025}.  The pure
quasifree state $\omega_{\rm f}$ satisfies twisted duality
\cite{Foit1983,BaumgartelJurkeLledo2002}, \cite[Cor.~15]{vanLuijk2025}.
Under these conditions Araki's construction identifies the spin
half-chain algebras with the fermionic ones and turns twisted duality into
Haag duality \cite[Props.~6.6 and~6.7]{KEYL_2006},
\cite[Remark~10]{vanLuijk2025}.  Purity follows: $\M_R^{\rm s}$ is a factor
and $\M_L^{\rm s}\vee\M_R^{\rm s}=\pi_{\rm s}(\A_{\rm s})''$, so
$\pi_{\rm s}(\A_{\rm s})'=(\M_L^{\rm s})'\cap(\M_R^{\rm s})'
=\M_R^{\rm s}\cap(\M_R^{\rm s})'=\bbC\Id$.
\end{proof}

\subsection{Interactions and ground states}

Let $h=M_{h(z)}$ be a translation-invariant one-particle Hamiltonian with
$\|h_r\|\le C_N(1+|r|)^{-N}$ for every $N$.  Its fermionic interaction is
$\Phi_{\rm f}(\{x\})=c_x^*h_0c_x$ and
$\Phi_{\rm f}([x,x+r])=c^*_{x+r}h_rc_x+c^*_xh_r^*c_{x+r}$ for $r\ge1$; these
terms are even, supported in the interval indicated, and of norm at most
$2\nu\|h_r\|$.  The spin interaction is
\begin{equation}\label{eq:spin-interaction}
    \Phi_{\rm s}(X):=\vartheta\bigl(\Phi_{\rm f}(X)\bigr).
\end{equation}

\begin{lemma}\label{lem:jw-locality}
$\Phi_{\rm s}$ is a translation-invariant even interaction with
$\Phi_{\rm s}(X)\in\A_{\rm s}(X)$, $\|\Phi_{\rm s}(X)\|=\|\Phi_{\rm f}(X)\|$
and $\sum_{X\ni0}(1+\diam X)^N\|\Phi_{\rm s}(X)\|<\infty$ for every $N$.
The derivation $\delta_{\rm s}(B)=i\sum_X[\Phi_{\rm s}(X),B]$ converges
for local $B$, and $\delta_{\rm s}(\vartheta(A))=\vartheta(\delta_h(A))$
for every even local $A$.  If $\|h_r\|\le Ce^{-\mu|r|}$, then
$\Phi_{\rm s}$ decays exponentially at the same rate.
\end{lemma}

\begin{proof}
The first statements follow from \cref{lem:jw-local}, since $\vartheta$ is
isometric and preserves supports of even interval operators, and from
$\sum_{X\ni0}(1+\diam X)^N\|\Phi_{\rm s}(X)\|\le2\nu\sum_{r\ge0}
(1+r)^{N+1}\|h_r\|$.  For even local $A$ and every $X$, both $A$ and
$\Phi_{\rm f}(X)$ lie in the even part of $\A_{\rm f}(J)$ for an interval
$J\supset X\cup\supp A$, on which $\vartheta=\vartheta_J$ is a
homomorphism, so
$[\Phi_{\rm s}(X),\vartheta(A)]=\vartheta([\Phi_{\rm f}(X),A])$; summing
over $X$ gives the intertwining relation.
\end{proof}

By \cref{lem:jw-locality} the fermionic and the spin ground-state
inequalities agree for even test operators.  Odd test operators are
controlled by translation invariance and clustering.

\begin{lemma}\label{lem:odd-tests}
Let $\Phi$ be a translation-invariant even interaction on $\A_{\rm s}$ with
$\sum_{X\ni0}\|\Phi(X)\|<\infty$, and let $\psi$ be an even,
translation-invariant factor state with $-i\psi(B^*\delta(B))\ge0$ for
every even local $B$.  Then $-i\psi(B^*\delta(B))\ge0$ for every local
$B$.
\end{lemma}

\begin{proof}
Write $B=B_++B_-$.  Since $\delta$ preserves the grading and $\psi$ is
even, the cross terms vanish, so it suffices to treat odd $B$.  Let
$B_n$ be the translate of $B$ by $n$ sites.  For $n$ large, $BB_n$ is even
and local, and $B$, $B^*$ commute with $B_n$, so
\begin{equation}
0\le-i\psi\bigl((BB_n)^*\delta(BB_n)\bigr)
=-i\psi\bigl(B^*\delta(B)\,B_n^*B_n\bigr)
-i\psi\bigl(B^*B\,B_n^*\delta(B_n)\bigr)+o(1),
\end{equation}
where the error comes from replacing $\delta(B)$ by local approximants and
tends to $0$ as $n\to\infty$.  A factor state on the quasi-local algebra
$\A_{\rm s}$ has the uniform clustering property
\cite{Powers1967,bratteli2013operator1}, and $\psi$ is translation
invariant, so the right-hand side converges to
$-2i\,\psi(B^*B)\,\psi(B^*\delta(B))$.  If $\psi(B^*B)>0$ the claim
follows.  If $\psi(B^*B)=0$, then $\pi_\psi(B)\Omega_\psi=0$ and
$\psi(B^*\delta(B))=\langle\pi_\psi(B)\Omega_\psi,
\pi_\psi(\delta(B))\Omega_\psi\rangle=0$.
\end{proof}

\begin{theorem}\label{thm:jw-ground}
Let $h=M_{h(z)}$ have superpolynomially decaying hopping, $\ker h=\{0\}$,
and spectrum accumulating at $0$ from both sides in the sense of
\cref{prop:ground-states}(ii).  Let $\omega_{\rm f}=\omega_P$ with
$P=\chi_{(-\infty,0)}(h)$, and assume that $\omega_{\rm s}$ is pure, as is
the case when $\M_R^{\rm f}$ is not of type I.  Then $\omega_{\rm s}$ is
the unique ground state of the spin interaction \cref{eq:spin-interaction}.
\end{theorem}

\begin{proof}
By \cref{prop:ground-states}(i), $\omega_{\rm f}$ is a ground state of
$\delta_h$.  For even local $B$, \cref{lem:jw-locality} gives
$-i\omega_{\rm s}(B^*\delta_{\rm s}(B))
=-i\omega_{\rm f}(A^*\delta_h(A))\ge0$ with $A=\vartheta^{-1}(B)$.  As
$\omega_{\rm s}$ is even, translation invariant and pure,
\cref{lem:odd-tests} shows that it is a ground state.

Let $\psi$ be any ground state of $\delta_{\rm s}$.  Since $\Phi_{\rm s}$
is even, $\psi\circ\Theta_{\rm s}$ is a ground state, and so is the even
state $\bar\psi=\frac12(\psi+\psi\circ\Theta_{\rm s})$.  The even state
$\bar\psi_{\rm f}:=\bar\psi\circ\vartheta\circ E_{\rm f}$ on $\A_{\rm f}$
satisfies $-i\bar\psi_{\rm f}(A^*\delta_h(A))\ge0$ for every even local
$A$ by \cref{lem:jw-locality}, so $\bar\psi_{\rm f}=\omega_{\rm f}$ by
\cref{prop:ground-states}(ii).  Hence $\bar\psi$ and $\omega_{\rm s}$ agree
on $\A_{\rm s}^+$, and being even, $\bar\psi=\omega_{\rm s}$.  Since
$\psi\le2\bar\psi$ and $\omega_{\rm s}$ is pure, $\psi=\omega_{\rm s}$.
\end{proof}

In the GNS representation of $\omega_{\rm s}$, the cyclic subspace of the
even algebra is the even sector of the fermionic Fock space, and by
\cref{lem:jw-locality} the spin dynamics acts there as the fermionic one.
Under the hypotheses of \cref{thm:jw-ground} the latter has spectrum
accumulating at $0$, created by particle--hole pairs of arbitrarily small
energy, so the spin ground state is gapless as well.

\subsection{Proof of \texorpdfstring{\cref{thm:main}}{the main theorem}}
\label{sec:proof-main}

Let $\nu=4$, so that $\A_{\rm s}=\bigotimes_\z M_{16}$.  Let
$\mathfrak F$ consist of the states $\omega_{\boldsymbol\kappa}$ with
$\boldsymbol\kappa\in[0,1]^2\setminus\{0\}$ and zeros as in
\cref{eq:standing}, and let $\mathfrak S$ be the set of their spin duals.
None of the states in $\mathfrak F$ has a type I half-chain algebra
(\cref{prop:modular-family-types}), so the states in $\mathfrak S$
are pure by \cref{lem:jw-halfchain}.

(i) By \cref{lem:jw-halfchain}, $\M_R(\omega_{\rm s})\cong
\M_R(\omega_{\rm f})$.  \Cref{cor:modular-family-examples}, applied with
$\boldsymbol\kappa=(1,0)$, $(\lambda,0)$ and $(e^{-1},e^{-\sqrt2})$, gives
the types $\mathrm{II}_\infty$, $\mathrm{III}_\lambda$ for $0<\lambda<1$,
and $\mathrm{III}_1$.  Half-chain algebras of spin chains are hyperfinite \cite{KEYL_2006}.

(ii) By \cref{thm:almost-local}, each $\omega_{\rm f}\in\mathfrak F$ is the negative
spectral projection state of a smooth gapless Bloch Hamiltonian satisfying
the hypotheses of \cref{thm:jw-ground}.  \Cref{lem:jw-locality,thm:jw-ground}
give the claim.

(iii) By \cref{lem:jw-entropy} the spin entropies equal the fermionic ones.
\Cref{cor:lacunary-examples,prop:polynomial,prop:slow}, applied with the
same mixing ratios as in (i), give the stated growth rates, and
\cref{cor:entropy-diverges} gives divergence for every state of
$\mathfrak S$.  \qed

\section{Discussion}\label{sec:discussion}

\Cref{thm:main} shows that the statement of Mohari \cite{Mohari2014} that
the half-chain restriction of a translation-invariant pure state is of
type I or of type III is false.  The main claim, stated as Theorem~3.3 in
\cite{Mohari2014}, is reduced there to Theorem~2.5.  The gap lies in
the passage from a trace on the algebra $\mathcal M^{1}(p)$ to a normal
trace on its weak closure $\mathcal M(p)$.  The unitary identifications (8)
and (9) of \cite{Mohari2014} are meant to justify this passage, but they
do not provide normality, and under the purity assumption one has
$\mathcal M(p)\cong\B(K)$, which carries no finite normal trace when $K$
is infinite dimensional.

Beyond the counterexample, the construction shows that the von Neumann
type and the growth of the interval entropy carry different information.
In our family the Blaschke zeros fix the length scales that enter the
entropy, while the mixing ratios fix the modular spectrum of the
half-chain and hence the type.  The two sets of parameters can be chosen
independently, so states of different types share the same entropy
growth, and type $\mathrm{III}_1$ states can have entropy diverging more
slowly than any prescribed function.  The entropy always diverges for
half-chain algebras that are not of type I (\cref{cor:entropy-diverges}),
in agreement with the fact that an area law implies the split property
\cite{matsui2011boundedness}.

The locality of the parent Hamiltonian is essential.  Our examples admit
gapless parents with superpolynomially decaying interactions, while gapped
quadratic parents must violate $\sum_r|r|\|h_r\|<\infty$
(\cref{prop:gapped-long-range}), and exponentially local quadratic
Hamiltonians with finitely many flavours obey the dichotomy of
\cref{thm:analytic-dichotomy}.  It remains possible that exponential
locality restores a much stronger dichotomy also for interacting systems,
as suggested by \cref{conj:short-range-dichotomy}.

Together with the exclusion of type $\mathrm{II}_1$
\cite{KEYL_2006,KEYL_2008}, our constructions show that translation
invariance and purity alone impose few restrictions on the half-chain
factor: type $\mathrm{II}_\infty$ and every type $\mathrm{III}_\lambda$,
$0<\lambda\le1$, occur.  It would be interesting to know which of these
possibilities survive stronger physical assumptions, such as finite-range
or exponentially decaying interactions, or frustration freeness.

\backmatter

\bmhead{Acknowledgements}
We thank Lauritz van Luijk, Henrik Wilming and Han Ma for inspiring
discussions.

\bmhead{Use of AI tools}
Generative AI tools, including ChatGPT (GPT-5.6 Sol) and Claude Opus 5.5, were used extensively throughout the development of this work to explore mathematical arguments, suggest proof strategies, and assist in drafting and revising portions of the manuscript. Some intermediate arguments and proof ideas emerged from AI-assisted discussions and were subsequently examined, refined, and further developed by the authors. The locality property of parent Hamiltonians and Conjecture \ref{conj:short-range-dichotomy} originated with the human authors. All mathematical statements and proofs presented in the manuscript have been independently verified by the authors, who take full responsibility for its contents. The figures were originally drawn by the authors and subsequently refined with the assistance of ChatGPT.

\section*{Declarations}

\begin{itemize}
\item \textbf{Funding.} Research at Perimeter Institute is supported in
part by the Government of Canada through the Department of Innovation,
Science and Economic Development and by the Province of Ontario through the
Ministry of Colleges, Universities and Research Excellence. RL is also
supported by the Simons Collaboration on Global Categorical Symmetries
through Simons Foundation grant 888996.
\item \textbf{Competing interests.} The authors have no relevant financial
or non-financial interests to disclose.
\item \textbf{Data availability.} Data sharing is not applicable to this
article, as no datasets were generated or analysed during the current
study.
\end{itemize}

\begin{appendices}

\section{Quasifree states}\label{app:quasifree}

Let $\hilb$ be a separable Hilbert space, $0\le C\le\Id$, and
$(\pi_C,\cH_C,\Omega_C)$ the GNS triple of $\omega_C$, with
$\M_C:=\pi_C(\CAR(\hilb))''$.  Standard references are
\cite{PowersStormer1970,Araki1970}; the classification of infinite tensor
products of type I factors (ITPFI factors) is due to Araki and Woods
\cite{ArakiWoods1968}.

\subsection{Basic properties}

\begin{proposition}\label{prop:qf-basic}
\begin{enumerate}[label=\textup{(\roman*)}]
    \item $\omega_C$ is a factor state, so $\M_C$ is a factor.
    \item $\omega_C$ is pure if and only if $C$ is a projection; then
    $\M_C=\B(\cH_C)$.
    \item The restriction of $\omega_C$ to $\CAR(\mathfrak k)$, for a closed
    subspace $\mathfrak k$, is $\omega_{P_{\mathfrak k}CP_{\mathfrak k}}$.
    \item $\omega_A$ and $\omega_B$ are quasi-equivalent if and only if
    $A^{1/2}-B^{1/2}$ and $(\Id-A)^{1/2}-(\Id-B)^{1/2}$ are Hilbert--Schmidt
    \cite[Thm.~5.1]{PowersStormer1970}.
    \item $\M_C$ is of type I if and only if $C-C^2$ is trace class.
    \item If $C$ has an orthonormal eigenbasis with eigenvalues $c_n$, then
    $\M_C$ is isomorphic to the ITPFI factor
    $\bigotimes_n\bigl(M_2,\diag(1-c_n,c_n)\bigr)$.
\end{enumerate}
\end{proposition}

\begin{proof}
(i), (ii), (iii) and (vi) are in \cite{PowersStormer1970,Araki1970}; for
(vi) see also the proof of \cite[Lemma~26]{vanLuijk2025}.  For (v), by
\cite[Lemma~5.3]{PowersStormer1970} the factor is of type I if and only if
$C=E+T$ with $E$ a spectral projection of $C$ and $T$ trace class.  If
$C=E+T$, then $C-C^2=T-ET-TE-T^2$ is trace class.  Conversely, with
$E=\chi_{[1/2,1]}(C)$ we have $E(\Id-C)\le2C(\Id-C)$ and
$(\Id-E)C\le2C(\Id-C)$, so $T=C-E$ is trace class if $C-C^2$ is.
\end{proof}

\begin{lemma}\label{lem:qf-entropy}
If $\dim\hilb<\infty$, the density matrix of $\omega_C$ has entropy
$\Tr\hb(C)$.
\end{lemma}

\begin{proof}
In an eigenbasis $(f_j)$ of $C$ with eigenvalues $c_j$, Wick's rule
\cref{eq:quasifree} shows that the density matrix is diagonal in the
occupation-number basis, with entries $\prod_j c_j^{n_j}(1-c_j)^{1-n_j}$,
$n_j\in\{0,1\}$.  Its entropy is $\sum_j\hb(c_j)$.
\end{proof}

\subsection{Block decompositions}

\begin{lemma}\label{lem:block-diagonal}
Let $\hilb=\hilb_1\oplus\hilb_2$ and $C=C_1\oplus C_2$, and let $\M_i$ be
the GNS factor of $\omega_{C_i}$ on $\CAR(\hilb_i)$.  Then $\omega_C$ is
the graded product of $\omega_{C_1}$ and $\omega_{C_2}$.  If the
implementer $\Gamma_1$ of the parity in the first GNS representation
belongs to $\M_1$, then $\M_C\cong\M_1\,\ovl\otimes\,\M_2$.  This is the
case if $C_1$ is a projection.
\end{lemma}

\begin{proof}
All mixed two-point functions vanish, so by \cref{eq:quasifree} $\omega_C$
is the graded product.  On $\cH_1\otimes\cH_2$ its GNS representation is
$c(f_1)\mapsto\pi_{C_1}(c(f_1))\otimes\Id$ and
$c(f_2)\mapsto\Gamma_1\otimes\pi_{C_2}(c(f_2))$.  If $\Gamma_1\in\M_1$,
then $\Id\otimes\pi_{C_2}(c(f_2))=(\Gamma_1\otimes\Id)(\Gamma_1\otimes
\pi_{C_2}(c(f_2)))$ lies in the generated von Neumann algebra, which
therefore equals $\M_1\,\ovl\otimes\,\M_2$.  If $C_1$ is a projection,
$\M_1=\B(\cH_1)$ by \cref{prop:qf-basic}(ii).
\end{proof}

\subsection{Constant covariance sectors}

\begin{lemma}\label{lem:half-trace}
If $\hilb_0$ is infinite dimensional, then $\omega_{\frac12\Id}$ is the
unique tracial state of $\CAR(\hilb_0)$, and its GNS factor is the
hyperfinite factor $R_{\mathrm{II}_1}$.
\end{lemma}

\begin{proof}
On the matrix algebra generated by the first $n$ modes of an orthonormal
basis the density matrix is $2^{-n}\Id_{2^n}$ by \cref{lem:qf-entropy},
and the union of these algebras is dense.  The GNS factor is
$\bigotimes_n(M_2,\frac12\tr)\cong R_{\mathrm{II}_1}$.
\end{proof}

\begin{proposition}\label{prop:constant-sectors}
Let $C=P\oplus\bigoplus_{j=1}^mt_j\Id_{\hilb_j}$ with $P$ a projection,
every $\hilb_j$ infinite dimensional, and $0<t_j<1$.  Let $\cH_P$ be the
GNS space of $\omega_P$.  Then
\begin{equation}
    \M_C\cong\B(\cH_P)\,\ovl\otimes\,\bigotimes_{j=1}^m\bigotimes_{n\in\bbN}
    \bigl(M_2,\diag(1-t_j,t_j)\bigr).
\end{equation}
If all $t_j=\frac12$ and $\cH_P$ is infinite dimensional, $\M_C$ is the
hyperfinite factor of type $\mathrm{II}_\infty$.  Otherwise let $G$ be the
closed subgroup of $\bbR_+^\times$ generated by the ratios
$r_j=\min\{t_j/(1-t_j),(1-t_j)/t_j\}$ with $t_j\ne\frac12$; then $\M_C$ is
of type $\mathrm{III}_\lambda$ if $G=\lambda^\z$ with $0<\lambda<1$, and of
type $\mathrm{III}_1$ if $G=\bbR_+^\times$.
\end{proposition}

\begin{proof}
The tensor product decomposition follows from
\cref{lem:block-diagonal} and \cref{prop:qf-basic}(vi).  If all $t_j=\frac12$,
the second factor is $R_{\mathrm{II}_1}$ by \cref{lem:half-trace}, and
$\B(\cH_P)\,\ovl\otimes\,R_{\mathrm{II}_1}$ is of type $\mathrm{II}_\infty$.
Otherwise the infinite tensor product contains each one-mode state
infinitely often, and its asymptotic ratio set is $\{0\}\cup G$
\cite{ArakiWoods1968}; tensoring with a type I factor does not change the
type.
\end{proof}

\begin{corollary}\label{cor:IIoo-sector}
If $C=P\oplus\frac12\Id_{\hilb_{\rm t}}$ with $P$ a projection on an
infinite-dimensional space and $\hilb_{\rm t}$ infinite dimensional, then
$\M_C$ is the hyperfinite factor of type $\mathrm{II}_\infty$.
\end{corollary}

\begin{proof}
The GNS space of a pure quasifree state on the CAR algebra of an
infinite-dimensional space is infinite dimensional, since that algebra is
simple and infinite dimensional.  Apply \cref{prop:constant-sectors}.
\end{proof}

\subsection{Ground states of quasifree dynamics}\label{app:groundstates}

Let $h$ be bounded and self-adjoint, $\tau$ the quasifree dynamics and
$\delta_h$ its generator, as in \cref{sec:ground-states}.  For
$\phi\in L^1(\bbR)$ put $\tilde\phi(\lambda)=\int\phi(t)e^{i\lambda t}\,
\rmd t$ and $\tau_\phi(B)=\int\phi(t)\tau_t(B)\,\rmd t$.  The Arveson
spectrum $\Sp_\tau(B)$ is the set of $\lambda\in\bbR$ such that
$\tau_\phi(B)\ne0$ for some $\phi$ with $\tilde\phi$ supported in any
given neighbourhood of $\lambda$; if $\tilde\phi$ vanishes on a
neighbourhood of $\Sp_\tau(B)$ then $\tau_\phi(B)=0$, and
$\Sp_\tau(B_1B_2)\subset\ovl{\Sp_\tau(B_1)+\Sp_\tau(B_2)}$
\cite[Sec.~3.2.3]{bratteli2013operator1}.  If $f\in\chi_K(h)\hilb$ for a
compact set $K$, then $\tau_\phi(c^*(f))=c^*(\tilde\phi(h)f)$ and, by the antilinearity
of $f\mapsto c(f)$, $\tau_\phi(c(f))=c\bigl(g_\phi(h)f\bigr)$ with
$g_\phi(\lambda):=\ovl{\tilde\phi(-\lambda)}$, so
\begin{equation}\label{eq:arveson-c}
    \Sp_\tau\bigl(c^*(f)\bigr)\subset K,
    \qquad
    \Sp_\tau\bigl(c(f)\bigr)\subset -K .
\end{equation}

\begin{lemma}\label{lem:fock}
Let $P$ be a projection on $\hilb$ and $\omega$ a state on $\CAR(\hilb)$
with $\omega(c^*(f)c(f))=0$ for $f\in\Ran(\Id-P)$ and
$\omega(c(g)c^*(g))=0$ for $g\in\Ran P$.  Then $\omega=\omega_P$.
\end{lemma}

\begin{proof}
In the GNS representation of $\omega$ the operators $c(f)$,
$f\in\Ran(\Id-P)$, and $c^*(g)$, $g\in\Ran P$, annihilate the cyclic
vector.  Together with their adjoints they generate $\CAR(\hilb)$ and obey
canonical anticommutation relations among themselves, the mixed
anticommutators vanishing.  By these relations every monomial is a linear
combination of normal-ordered monomials, with all annihilating factors on
the right, whose coefficients depend only on $P$ and on the inner
products of the arguments.  The expectation of a normal-ordered monomial
other than $\Id$ vanishes, so $\omega$ is determined by $P$.  The state
$\omega_P$ has the same properties by \cref{eq:quasifree}; compare the
characterization of Fock states in \cite{Araki1970}.
\end{proof}

\begin{lemma}\label{lem:arveson}
Let $\omega$ be a ground state of a $C^*$-dynamical system
$(\mathfrak A,\tau)$ and let $B\in\mathfrak A$ with
$\Sp_\tau(B)\subset(-\infty,-\epsilon]$ for some $\epsilon>0$.  Then
$\omega(B^*B)=0$.
\end{lemma}

\begin{proof}
In the GNS representation $\tau$ is implemented by $e^{itH}$ with $H\ge0$
and $H\Omega=0$ \cite[Prop.~5.3.19]{bratteli2013operator2}, and
$\pi(\tau_\phi(B))\Omega=\tilde\phi(H)\pi(B)\Omega$.  Choose $\phi_R$ with
$\tilde\phi_R\in C_c^\infty(\bbR)$, $0\le\tilde\phi_R\le1$,
$\tilde\phi_R=0$ on $(-\infty,-\epsilon/2]$ and $\tilde\phi_R=1$ on
$[0,R]$.  Then $\tau_{\phi_R}(B)=0$, so $\tilde\phi_R(H)\pi(B)\Omega=0$,
and $\tilde\phi_R(H)\to\Id$ strongly as $R\to\infty$ because $H\ge0$.
\end{proof}

\begin{proof}[Proof of \cref{prop:ground-states}]
(i) In the Fock representation of $\omega_P$ the dynamics is implemented
by the second quantization of $|h|$ after the particle--hole
transformation on $\Ran P$, which is positive and annihilates the vacuum,
so $\omega_P$ is a ground state
\cite[Prop.~5.3.19]{bratteli2013operator2}; see also
\cite[Sec.~5.1]{vanLuijk2025}.  Let $\omega$ be any ground state.  For
$0<\epsilon<M$ and $f\in\chi_{[\epsilon,M]}(h)\hilb$,
\cref{eq:arveson-c} and \cref{lem:arveson} give $\omega(c^*(f)c(f))=0$.
These $f$ are dense in $\chi_{(0,\infty)}(h)\hilb=\Ran(\Id-P)$, the
equality because $\ker h=\{0\}$.  Likewise $\omega(c(f)c^*(f))=0$ for
$f\in\Ran P$, and \cref{lem:fock} gives $\omega=\omega_P$.

(ii) Let $\omega$ be even with $-i\omega(A^*\delta_h(A))\ge0$ for even
$A\in\A_{\rm loc}$.  The even local elements are a core for the generator
of $\tau$ restricted to the even subalgebra $\CAR(\hilb)^+$, as in
\cref{sec:ground-states}, so $\omega|_{\CAR(\hilb)^+}$ is a ground state
of $(\CAR(\hilb)^+,\tau)$.  Let $0<\epsilon<M$,
$f\in\chi_{[2\epsilon,M]}(h)\hilb$ and $g\in\chi_{(0,\epsilon)}(h)\hilb$.
By \cref{eq:arveson-c}, the even elements $c^*(g)c(f)$ and $c(g)c(f)$ have
Arveson spectrum in $(-\infty,-\epsilon]$, so \cref{lem:arveson} shows that
$v=\pi_\omega(c(f))\Omega_\omega$ is annihilated by
$\pi_\omega(c^*(g))$ and by $\pi_\omega(c(g))$.  Choosing $g\ne0$, which
is possible by hypothesis, the relation $\{c(g),c^*(g)\}=\|g\|^2\Id$ gives
$v=0$, so $\omega(c^*(f)c(f))=0$.  The same argument with
$f\in\chi_{[-M,-2\epsilon]}(h)\hilb$, $g\in\chi_{(-\epsilon,0)}(h)\hilb$
and the elements $c(g)c^*(f)$, $c^*(g)c^*(f)$ gives
$\omega(c(f)c^*(f))=0$.  As in (i), \cref{lem:fock} gives
$\omega=\omega_P$.
\end{proof}

\section{An analytic toolbox}\label{app:analysis}

Standard references are \cite{Katznelson2004,Duren1970,Garnett2007,
GarciaMashreghiRoss2016}.

\subsection{Fourier transform and Hardy space}\label{app:fourier}

The Fourier transform $(\F f)(z)=\sum_xf(x)z^x$ is a unitary
$\ell^2(\z)\to L^2(\bbT)$ by Parseval's identity, with $\F\delta_x=z^x$; in
the lattice terminology $\bbT$ is the Brillouin zone and $z=e^{ik}$ with
crystal momentum $k$.  The translation $V\delta_x=\delta_{x+1}$ becomes
$M_z$.  A bounded operator commuting with $M_z\otimes\Id_{\bbC^\nu}$ is a
multiplication operator by an element of $L^\infty(\bbT,M_\nu)$, by the
argument in the proof of \cref{prop:one-species} applied to each matrix
entry.  The Fourier transform maps $\ell^2(\bbN)$ onto the Hardy space
$H^2$, the space of $f=\sum_{n\ge0}c_nz^n$ with $\sum|c_n|^2<\infty$,
equivalently of $L^2$ boundary values of such power series on $\bbD$, and
it turns the unilateral shift into $M_z|_{H^2}$.  We write $\Hinf$ for
the bounded analytic functions on $\bbD$, identified with their boundary
values, and $\cHinf$ for their complex conjugates; note
$(H^2)^\perp=\ovl{zH^2}$.

\subsection{Inner functions, Blaschke products and model spaces}
\label{app:blaschke}

A function $\theta\in\Hinf$ is inner if $|\theta|=1$ almost everywhere on
$\bbT$.  For $\{a_n\}\subset\bbD\setminus\{0\}$ with
$\sum_n(1-|a_n|)<\infty$, an integer $k\ge0$ and $\zeta\in\bbT$, the
Blaschke product
\begin{equation}
    B(z)=\zeta z^k\prod_{n}\frac{|a_n|}{a_n}\,\frac{a_n-z}{1-\ovl{a_n}z}
\end{equation}
converges uniformly on compact subsets of $\bbD$ and is an inner function
whose zeros are $0$ with multiplicity $k$ and the points $a_n$, counted
with multiplicity \cite{Duren1970,Garnett2007}.  For real zeros
$0\le a_n<1$ this is the product \cref{eq:standing} up to a sign, with
$b_0(z)=-z$.  The degree of a finite Blaschke product is its number of
zeros.  For an inner function $\theta$ the model space is
$K_\theta:=H^2\ominus\theta H^2$.

\begin{proposition}\label{prop:model-spaces}
Let $\theta,\psi,\chi$ be inner.
\begin{enumerate}[label=\textup{(\roman*)}]
    \item The closed $M_z$-invariant subspaces of $H^2$ other than $\{0\}$
    are exactly the spaces $\theta H^2$ \cite{Beurling1949}.
    \item $T_\theta$ is an isometry, $T_\theta T_\theta^*$ is the
    projection onto $\theta H^2$, and $K_\theta=\ker T_\theta^*$.
    \item $K_{\psi\theta}=K_\psi\oplus\psi K_\theta$, and $K_\theta\subset
    K_{\theta\chi}$.
    \item $g\in K_\theta$ if and only if $g\in H^2$ and
    $\bar\theta g\in\ovl{zH^2}$.
    \item For a Blaschke product $B$, $\dim K_B$ is the number of zeros of
    $B$ counted with multiplicity.
\end{enumerate}
\end{proposition}

\begin{proof}
(ii) Multiplication by $\theta$ is isometric on $H^2$ since $|\theta|=1$,
so its range $\theta H^2$ is closed and $T_\theta T_\theta^*$ is the
projection onto it.  (iii) $H^2=K_\psi\oplus\psi H^2=K_\psi\oplus\psi(K_\theta
\oplus\theta H^2)$, and $\theta\chi H^2\subset\theta H^2$.  (iv)
$\langle\theta f,g\rangle=\langle f,\bar\theta g\rangle$ for $f\in H^2$.
(v) For $a\in\bbD$, $b_aH^2$ is the space of $f\in H^2$ with $f(a)=0$,
which has codimension one, so $\dim K_{b_a}=1$; by (iii) $\dim K_B$ equals
the degree of a finite Blaschke product.  If $B$ has infinitely many
zeros, every finite partial product $B_N$ divides $B$, so
$\dim K_B\ge\dim K_{B_N}=N$ for all $N$ by (iii) \cite{GarciaMashreghiRoss2016}.
\end{proof}

\subsection{Derivative bounds for Blaschke products}

\begin{lemma}\label{lem:blaschke-derivatives}
Let $B$ be a Blaschke product with real zeros $a_n\in[0,1)$, counted with
multiplicity, and $\Delta:=\sum_n(1-a_n)<\infty$.  Then $B$ is smooth on
$\bbT\setminus\{1\}$, and for every $r\ge1$ there is a constant $C_r$,
independent of $B$, with
\begin{equation}\label{eq:blaschke-derivatives}
    \bigl|\partial_\theta^rB(e^{i\theta})\bigr|
    \le C_r(1+\Delta)^r|\theta|^{-2r},
    \qquad0<|\theta|\le\pi .
\end{equation}
\end{lemma}

\begin{proof}
For $|z|=1$ a direct computation gives
$\partial_\theta\log b_a(e^{i\theta})=iP_a(\theta)$ with
$P_a(\theta)=(1-a^2)/|e^{i\theta}-a|^2$.  Since
$|e^{i\theta}-a|^2=(1-a)^2+4a\sin^2(\theta/2)$, we have
$|e^{i\theta}-a|\ge|\theta|/(2\pi)$ for all $a\in[0,1)$.  Writing
$|e^{i\theta}-a|^{-2}=(e^{i\theta}-a)^{-1}(e^{-i\theta}-a)^{-1}$ and
differentiating gives $|\partial_\theta^jP_a(\theta)|\le C_j(1-a)
|\theta|^{-2-j}$ for $j\ge0$.  The zeros accumulate only at $1$, so on
$\{|\theta|\ge t\}$ the product converges with all derivatives, and
$B(e^{i\theta})=\pm e^{iw(\theta)}$ with $w'=\sum_nP_{a_n}$, hence
$|w^{(j)}(\theta)|\le C_j\Delta|\theta|^{-1-j}$ for $j\ge1$.  By the Faà di
Bruno formula \cite[p.~139]{Comtet1974}, $\partial_\theta^rB$ is $B$ times
a sum of products $\prod_j(iw^{(j)})^{m_j}$ with $\sum_jjm_j=r$, each
bounded by $C\Delta^{\sum m_j}|\theta|^{-r-\sum m_j}\le C(1+\Delta)^r
|\theta|^{-2r}$ for $|\theta|\le\pi$.
\end{proof}

\subsection{Fourier coefficients of infinite Blaschke products}

\begin{lemma}\label{lem:blaschke_long_range}
Let $u=\sum_{n\ge0}u_nz^n$ be a Blaschke product with infinitely many
zeros $a_n\to1$.  Then $u$ has no continuous extension to
$\bbD\cup\{1\}$, and $\sum_n|u_n|=\infty$.
\end{lemma}

\begin{proof}
A continuous extension would satisfy $u(1)=\lim u(a_n)=0$.  But the radial
limits of $u$ have modulus one almost everywhere, so there are points
$\zeta_j\to1$ on $\bbT$ and radii $r_j\to1$ with $|u(r_j\zeta_j)|>\frac12$,
a contradiction.  If $\sum_n|u_n|<\infty$, the power series would converge
uniformly on the closed disk and give such an extension.
\end{proof}

\end{appendices}

\bibliography{references}

\end{document}